\documentclass[11pt,a4paper]{article}

\usepackage[T1]{fontenc}
\usepackage{lmodern}
\usepackage{fullpage}
\usepackage{authblk}
\usepackage{amssymb}
\usepackage[english]{babel}
\usepackage[utf8]{inputenc}
\usepackage{enumerate}
\usepackage{cancel}
\usepackage{xspace}
\usepackage{bm}
\usepackage{algorithm, algorithmic}
\usepackage{float}
\usepackage{stmaryrd}
\usepackage[normalem]{ulem}
\usepackage{microtype}
\usepackage{amsthm}
\usepackage{amsmath}
\usepackage{amsfonts}
\usepackage{graphicx}
\usepackage{mathtools}
\usepackage{thmtools}

\usepackage{color}
\definecolor{darkgreen}{rgb}{0,0.5,0}
\definecolor{darkblue}{rgb}{0,0,0.8}
\definecolor{darkred}{rgb}{0.8,0,0}

 \usepackage{hyperref}
 \hypersetup{
   unicode=false,          % non-Latin characters in Acrobat’s bookmarks
   colorlinks=true,        % false: boxed links; true: colored links
   linkcolor=darkred,          % color of internal links (change box color with linkbordercolor)
   citecolor=darkblue,        % color of links to bibliography
   filecolor=magenta,      % color of file links
   urlcolor=blue           % color of external links
 }
\newtheorem{definition}{Definition}[section]
\newtheorem{invariant}[definition]{Invariant}
\newtheorem{lemma}[definition]{Lemma}
\newtheorem{theorem}[definition]{Theorem}
\newtheorem{problem}[definition]{Problem}

\graphicspath{{./figures/}}

\newcommand{\bigo}{\mathcal{O}}

\newcommand{\lca}{\mathsf{lca}}
\newcommand{\cP}{{\cal P}}
\newcommand{\cQ}{{\cal Q}}
\newcommand{\MOVE}{\textsf{MOVE}}
\newcommand{\DEL}{\textsf{DEL}}
\newcommand{\procinsert}{\textsf{insert}}
\newcommand{\procdelete}{\textsf{delete}}
\newcommand{\procmove}{\textsf{move}}
\newcommand{\nextrm}{\textrm{next}}
\newcommand{\prevrm}{\textrm{prev}}
\newcommand{\degree}{\textsf{deg}}

 \title{\bf\LARGE All-Pairs LCA in DAGs:\\ Breaking through the $O(n^{2.5})$ barrier\footnotemark[0]\footnotetext[0]{Fabrizio Grandoni is partially supported by the SNF Excellence Grant 200020B 182865/1. Giuseppe F. Italiano is partially supported by MUR, the Italian Ministry for University and Research, under PRIN Project AHeAD (Efficient Algorithms for HArnessing Networked Data). Aleksander Łukasiewicz and  Przemys\l{}aw~Uzna\'nski are Supported by the Polish National Science Centre grant 2019/33/B/ST6/00298.}}

\author[1]{\LARGE Fabrizio Grandoni}
 \author[2]{Giuseppe F. Italiano}
 \author[3]{Aleksander~Łukasiewicz}
 \author[4]{Nikos Parotsidis}
 \author[3]{Przemys\l{}aw~Uzna\'nski}
 \affil[1]{\large IDSIA, Lugano, Switzerland}
 \affil[2]{\large LUISS University, Rome, Italy}
\affil[3]{\large University of Wrocław, Wrocław, Poland}
 \affil[4]{\large Google Research, Zurich, Switzerland}
 
 \date{}
\begin{document}
\maketitle

% % \setcounter{page}{0}
 \thispagestyle{empty}

% % \setcounter{page}{0}
% \thispagestyle{empty}

%\linenumbers
\begin{abstract}
Let $G=(V,E)$ be an $n$-vertex directed acyclic graph (DAG). A \emph{lowest common ancestor} (LCA) of two vertices $u$ and $v$ is a common ancestor $w$ of $u$ and $v$ such that no descendant of $w$ has the same property.  
In this paper, we consider the problem of computing an LCA, if any, for all pairs of vertices in a DAG. 
The fastest known algorithms for this problem exploit fast matrix multiplication subroutines and have running times ranging from $\bigo( n^{2.687})$ [Bender et al.~SODA'01] down to $\bigo(n^{2.615})$ [Kowaluk and Lingas~ICALP'05] and $\bigo(n^{2.569})$ [Czumaj et al.~TCS'07]. Somewhat surprisingly, all those bounds would still be $\Omega(n^{2.5})$ even if  matrix multiplication could be solved optimally (i.e., $\omega=2$).  
This appears to be an inherent barrier for all the currently known approaches, which raises the natural question on whether one could break through the $\bigo(n^{2.5})$ barrier for this problem.  

In this paper, we answer this question affirmatively: in particular, we  present an   
$\widetilde\bigo(n^{2.447})$ ($\widetilde\bigo(n^{7/3})$ for $\omega=2$) algorithm for finding an LCA for all pairs of vertices in a DAG, which represents the first improvement on the running times for this problem in the last 13 years. 
A key tool in our approach is a fast algorithm to partition the vertex set of the transitive closure of $G$ into a collection of $\bigo(\ell)$ chains and $\bigo(n/\ell)$ antichains, for a given parameter $\ell$. 
As usual, a chain is a path while an antichain is an independent set. 
We then find, for all pairs of vertices, a \emph{candidate} LCA among the chain and antichain vertices, separately. 
The first set is obtained via a reduction to $(\max,\min)$ matrix multiplication. 
The computation of the second set can be reduced to Boolean matrix multiplication similarly to previous results on this problem. 
We finally combine the two solutions together in a careful (non-obvious) manner. 
\end{abstract}

\newpage

\setcounter{page}{1}

\section{Introduction}

Let $G=(V,E)$ be a directed acyclic graph (DAG), with $m$ edges and $n$ vertices. Let $u$ and $v$ be any two vertices in $G$: if there is a path from $u$ to $v$, we say that $u$ is an \emph{ancestor} of $v$ and that $v$ is a \emph{descendant} of $u$.  If $u$ is an ancestor of $v$ and $u\neq v$, we say that $u$ is a \emph{proper ancestor} of $v$
(and $v$ is a \emph{proper descendant} of $u$).  A \emph{lowest common ancestor} (LCA) of $u$ and $v$
 is the lowest (i.e., deepest) vertex $w$ that is an ancestor of both $u$ and $v$, i.e., no proper descendant of $w$ is an ancestor of both $u$ and $v$. In the special case of a tree, the lowest common ancestor of two vertices is always defined and is unique. In a DAG $G$, the existence of an LCA for a pair of vertices is not even guaranteed, and a pair of vertices can have as many as $(n-2)$ LCAs, where $n$ is the total number of vertices in $G$.

In this paper, we consider the problem of computing an LCA for all pairs of vertices in a DAG, which we refer to as the \emph{All-Pairs LCA problem}. This is a fundamental problem and has many important applications, including inheritance in object-oriented programming languages, analysis of genealogical data and modeling the behavior of complex systems in distributed computing  (see, e.g.,~\cite{DBLP:journals/jal/BenderFPSS05,Cottingham,Shaffer} for a list of applications and especially~\cite{DBLP:journals/jal/BenderFPSS05} for further references). 

The All-Pairs LCA problem for DAGs has been investigated in the last two decades, and many algorithms have been presented in the literature (see, e.g.,~\cite{DBLP:journals/jal/BenderFPSS05,DBLP:conf/soda/BenderPSS01,DBLP:journals/tcs/CzumajKL07,DBLP:conf/icalp/KowalukL05,DBLP:conf/esa/KowalukL07,DBLP:conf/swat/KowalukLN08}).
The problem was first considered by Bender et al.~\cite{DBLP:journals/jal/BenderFPSS05,DBLP:conf/soda/BenderPSS01}, who 
proved an  $\Omega(n^\omega)$ lower bound, by giving a
reduction from the transitive closure problem, and presented an algorithm that runs in 
  $\bigo(n^{(\omega+3)/2})$ time, where $\omega$ is the exponent of the fastest known matrix multiplication algorithm. Later on, Kowaluk and Lingas~\cite{DBLP:conf/icalp/KowalukL05}  improved this bound to 
$\bigo(n^{2+1/(4-\omega)})$ by showing that the All-Pairs LCA problem can be reduced to finding maximum witnesses for Boolean matrix multiplication and by providing an efficient solution to the latter problem. The current best bound  for the All-Pairs LCA problem is
%and with fastest being 
 $\bigo(n^{2.5286})$ by Czumaj et al.~\cite{DBLP:journals/tcs/CzumajKL07}. To achieve this bound, they solved the problem of finding maximum witnesses for Boolean matrix multiplication in time $\bigo(n^{2+\lambda})$, where $\lambda$ satisfies the equation $\omega(1,\lambda,1) = 1+2\lambda$. Here $\omega(1,x,1)$ is the exponent of the (rectangular) multiplication of an $n\times n^{x}$ matrix by an $n^{x}\times n$ matrix. The currently best known bound on  $\omega(1,x,1)$ implies a bound of $\bigo(n^{2.5286})$ for the All-Pairs LCA problem.

Somewhat surprisingly, all the currently known bounds for the All-Pairs LCA problem~\cite{DBLP:journals/jal/BenderFPSS05,DBLP:conf/soda/BenderPSS01,DBLP:journals/tcs/CzumajKL07,DBLP:conf/icalp/KowalukL05} would still be $\Omega(n^{2.5})$ even if  matrix multiplication could be solved optimally (i.e., $\omega=2$).  
This appears to be an inherent barrier for all the currently known approaches, which raises the natural question on whether one could break through the $\bigo(n^{2.5})$ barrier for this problem.

\paragraph{Our result.} In this paper we answer this question affirmatively by presenting a new algorithm which runs in time $\widetilde\bigo(n^{2.447})$. This is the first improvement on the running time for this problem in the last 13 years.
To this end we introduce some novel techniques, which differ substantially from previous approaches for the same problem.
In particular, we develop a new technique for covering a DAG $G$ with a small number of chains and antichains, which might also be of independent interest. Here, a chain is just a path in $G$, while an antichain is an independent set (i.e., a subset of vertices such that there is no edge between any two of them). In more detail, given a parameter $\ell \le n$, we show how to partition the vertices of $G$ into at most $\ell$ chains and $2n/\ell$ antichains in time $\bigo(n^2)$. We refer to this as an $(\ell,2n/\ell)$-decomposition of $G$.

We now sketch how to exploit this decomposition in order to compute efficiently all-pairs LCAs. Let $G^T$ be the transitive closure of $G$, which can be computed in time $\bigo(n^{\omega})$. Note that we can solve the All-Pairs LCA problem in $G$ by solving the same problem in $G^T$.
We first find an $(n^x,2{n}^{1-x})$-decomposition of $G^T$ in time $\bigo(n^2)$ for a parameter $x$, $0\leq x\leq 1$, to be fixed later. Next, for each pair of vertices, we find a candidate LCA among the chain and antichain vertices, separately. The first set can be obtained  in time $\widetilde{\bigo}(n^{\frac{\omega(1,x,1) + 2 + x}{2}})$ via a reduction to $(\max,\min)$ matrix multiplication, similarly to Bender et al.~\cite{DBLP:journals/jal/BenderFPSS05,DBLP:conf/soda/BenderPSS01} and to Czumaj et al.~\cite{DBLP:journals/tcs/CzumajKL07}.
The computation of the second set can be reduced to Boolean matrix multiplication in time $\widetilde{\bigo}(n^{1-x+\omega(1,x,1)})$.  
Combining the two solutions is non-trivial and requires some extra care.  
Putting all pieces together yields a total running time of   $\widetilde\bigo(n^{\omega} + n^{\frac{\omega(1,x,1)+2+x}{2}}+n^{1-x + \omega(1,x,1)})$. Balancing the last two terms implies $\omega(1,x,1)=3x$, and thus a running time of
$\widetilde\bigo(n^{\omega} + n^{1+2x})$. Using the equation $\omega(1,x,1)=3x$ to 
 tune the parameter $x$, yields $x = 0.7232761$, and hence a total running time of $\bigo(n^{2.447})$.

We remark that if  matrix multiplication could be solved optimally (i.e., $\omega=2$), several graph algorithms based on fast matrix multiplication would take either time $\widetilde\bigo(n^{2})$ or time $\widetilde\bigo(n^{2.5})$. As it was already mentioned, the previous algorithms by Bender et al.~\cite{DBLP:journals/jal/BenderFPSS05,DBLP:conf/soda/BenderPSS01}, by Kowaluk and Lingas~\cite{DBLP:conf/icalp/KowalukL05} and by Czumaj et al.~\cite{DBLP:journals/tcs/CzumajKL07} would all take time $\widetilde\bigo(n^{2.5})$. On the other side, under the same assumption, 
the running time of our algorithm would be $\widetilde\bigo(n^{2+\frac{1}{3}})$.  Thus, our improvement 
 suggests a possible separation between the All-Pairs LCA and the minimum / maximum witness for Boolean matrix multiplication (used by Czumaj et al.~\cite{DBLP:journals/tcs/CzumajKL07} as a reduction in their algorithm for All-Pairs LCA).

\paragraph{Related work.}

The problem of finding LCAs in trees was first  introduced by Aho et al.~\cite{DBLP:conf/stoc/AhoHU73}. 
The first optimal (linear preprocessing and $\bigo(1)$ time per query) solution to this problem was presented by Harel and Tarjan \cite{DBLP:journals/siamcomp/HarelT84}, although with a sophisticated data structure which is not practical. 
The first simple, near-optimal algorithm for LCAs in trees was introduced by Bender and Farach-Colton \cite{DBLP:conf/latin/BenderF00}. 
We remark that LCA problem in trees exemplifies a rather different structure than in DAGs.

Matrix multiplication is a fundamental problem, with a long line of algebraic approaches, with recent results by Stothers \cite{stothers2010complexity}, Vassilevska-Williams \cite{DBLP:conf/stoc/Williams12}, Le Gall \cite{DBLP:conf/issac/Gall14}, and finally \cite{DBLP:journals/corr/abs-2010-05846} which yielded $\omega < 2.37286$. 
There are known faster (under the assumption that $\omega>2$) algorithms for rectangular matrix multiplication, with current best bounds by Le Gall and Urrutia \cite{DBLP:conf/soda/GallU18}.
There is a long list of problems which are equivalent to matrix multiplication e.g., Boolean matrix multiplication witnesses with $\widetilde\bigo(n^\omega)$ \cite{DBLP:conf/focs/AlonGMN92} and All-Pairs Shortest Paths (APSP) in undirected unweighted graphs \cite{DBLP:journals/jcss/AlonGM97}.

There is a large family of graph and geometric problems for which the best known algorithms use matrix-multiplication and their complexity is between $n^\omega$ and $n^3$. Such problems include  All-Pair Bottleneck Paths (APBP): \cite{DBLP:conf/soda/DuanP09,DBLP:conf/stoc/VassilevskaWY07}, vertex APBP \cite{DBLP:conf/soda/ShapiraYZ07}, unweighted directed APSP \cite{DBLP:journals/jacm/Zwick02}, All-Pair Nondecreasing Paths \cite{DBLP:conf/icalp/DuanGZ18,DBLP:conf/icalp/DuanJW19,DBLP:journals/talg/Williams10}, and Dominance-, Hamming- and $L_1$- matrix products: \cite{DBLP:conf/icalp/IndykLLP04,DBLP:conf/cocoon/MinKZ09,DBLP:conf/soda/Yuster09}. Interestingly, for all the aforementioned problems, the best known algorithms would be of complexity $\widetilde\bigo(n^{2.5})$ if $\omega=2$.
 For fine-grained complexity of intermediate complexity problems and the relations between them, see recent results \cite{DBLP:journals/corr/abs-1911-06132,DBLP:conf/stoc/BringmannKW19,DBLP:conf/soda/DurajK0W20,DBLP:conf/innovations/Lincoln0W20}.

For other results on All-Pairs LCA, see \cite{DBLP:conf/esa/KowalukL07} on finding unique LCA. Other work in the area include~\cite{DBLP:journals/tcs/DashSHC13}.

For related problems of minimal witnesses of Boolean matrix multiplication, we refer to an algorithm for sparse matrices \cite{DBLP:journals/algorithmica/CohenY14}, and to a recent quantum algorithm \cite{DBLP:journals/corr/abs-2004-14064}. 
We also refer to \cite{DBLP:conf/swat/KowalukLN08} which introduced path covering techniques in the All-Pairs LCA problem. 
The authors observe that covering a DAG with a small number of paths might lead to faster algorithms, which is one of the key observations used in our algorithm. However, this observation alone does not imply a faster~algorithm.

The decomposition of a partially ordered set into disjoint chains and antichains can be seen as a special case of finding a \emph{cocoloring} of a graph. 
A \emph{cocoloring} of a graph is a partition of its vertices into cliques and independent sets.
The \emph{cochromatic number} of a graph is the cardinality of the smallest cocoloring.
This problem has been originally studied by Lesniak and Straight \cite{lesniak1977cochromatic}.
The special case of partitioning a sequence into monotonic subsequences has received considerable attention \cite{DBLP:journals/acta/Bar-YehudaF98, DBLP:journals/eik/BrandstadtK86, DBLP:journals/jgt/ErdosGK91, DBLP:journals/ipl/FominKN02, DBLP:journals/eik/Wagner84, DBLP:conf/tamc/YangCLZ07}, since 
it has many applications, including  book embeddings \cite{DBLP:journals/acta/Bar-YehudaF98}, and geometric algorithms \cite{DBLP:conf/isaac/ArroyueloCDDHLMNSS09, DBLP:journals/acta/Bar-YehudaF98, DBLP:conf/spire/ClaudeMN10}.

\section{Preliminaries}

Let  $G=(V,E)$ be a DAG, with $m$ edges and $n$ vertices. Without loss of generality we assume that $G$ is weakly connected (hence, $m\geq n-1$). If $(u,v) \in E(G)$ we say that $u$ is a \emph{parent} of $v$ and $v$ is a child of $u$. If there is a path from $u$ to $v$ in $G$ we say that $u$ is an \emph{ancestor} of $v$ and that $v$ is a \emph{descendant} of $u$.  If $u$ is an ancestor of $v$ and $u\neq v$, we say that $u$ is a \emph{proper ancestor} of $v$ (and $v$ is a \emph{proper descendant} of $u$).  A \emph{lowest common ancestor} (LCA) of $u$ and $v$
 is the lowest (i.e., deepest) vertex $w$ that is an ancestor of both $u$ and $v$, i.e., no proper descendant of $w$ is an ancestor of both $u$ and $v$. We use $LCA(u,v)$ to denote the set of LCAs of $u$ and $v$. In case there is no common ancestor of $u$ and $v$, $LCA(u,v)=\emptyset$. In this paper, we consider the following~problem.

\begin{problem}[All-Pairs LCA]
Let $G = (V,E)$ be a DAG. Compute a lowest common ancestor for all pairs of vertices $u,v \in V$.
\end{problem}

\paragraph*{Matrix multiplication.}
We use  $\textrm{MM}(X,Y,Z)$ to denote the time complexity of multiplying two matrices of dimensions $X \times Y$ and $Y \times Z$ respectively. We denote by $\omega$ the exponent of the fastest known matrix multiplication algorithm, 
i.e., $\textrm{MM}(n,n,n) = \bigo(n^\omega)$. The current best bound for $\omega$ is $\omega < 2.3728639$ \cite{DBLP:conf/issac/Gall14}.
We denote by $\omega(a,b,c)$ the rectangular matrix multiplication exponent, i.e.,  $\textrm{MM}(n^a,n^b,n^c) = \bigo(n^{\omega(a,b,c)})$. The following is a standard bound derived from reducing rectangular matrix multiplication to square matrix multiplication:
\begin{equation}
\label{eq:boundsquare}
\omega(1,x,1) \le 2 + x (\omega-2)\qquad\textrm{ for } 0 \le x \le 1.
\end{equation}
We introduce the following definition:
\begin{definition}
Let $\alpha>0.31389$ be the maximum value satisfying $\omega(1,\alpha,1) = 2$, and let $\beta = \frac{\omega-2}{1-\alpha}$.
\end{definition}
The following bound holds:
\begin{equation}
\label{eq:boundrect}
\omega(1,x,1) \le \begin{cases}2 + \beta(x-\alpha) \textrm{ when } \alpha \le x \le 1,\\ 2 \textrm{ when }0 \le x \le \alpha.\end{cases}
\end{equation}

We remark that there are better bounds on $\omega(1,x,1)$ (see, e.g.,  \cite{DBLP:conf/soda/GallU18}). In particular, the  following bound is known:
\begin{equation}
\label{eq:legallurrutia}
\omega(1,x,1) \le 1.690383 + 0.66288 \cdot x\textrm{ for } 0.7 \le x \le 0.75
\end{equation}

%\onote{This last bound is obtained by linear interpolation of bounds by Le Gall. I think that the rounding of the interpolation might be wrong and this numbers should actually be bigger on fifth-sixth decimal place. This is probably of minor importance, tough.}
%\paragraph*{Algorithm complexity.} 
We now sketch how all those bounds influence the $\widetilde\bigo(n^{\omega} + n^{1+2x})$ running time of our algorithm.
If we simply use square matrix multiplication as a subroutine to implement rectangular matrix multiplication (i.e., the bound in \eqref{eq:boundsquare}), combining this with equation $3x = \omega(1, x, 1)$, we obtain $x = \frac{2}{5-\omega}$. %and $\gamma=\frac{9-\omega}{5-\omega}$. %Notice that $\frac{9-\omega}{5-\omega}> \omega$ since $\omega< 3$. 
In this case, the running time of our algorithm would be $\widetilde{\bigo}(n^{\frac{9-\omega}{5-\omega}})\in \bigo(n^{2.522571})$, which is already an improvement over the algorithm by Czumaj et al.~\cite{DBLP:journals/tcs/CzumajKL07}. 
This bound can be further improved by using more sophisticated 
 rectangular matrix multiplication algorithms. In particular, using the bound in \eqref{eq:boundrect}, we get $x = \frac{2-\omega\alpha}{5-\omega-3\alpha}$ 
% $\gamma = \frac{2(2-\omega\alpha)}{5-3\alpha-\omega}+1$, 
and the running time of our algorithm becomes $\bigo(n^{2.489418})$, which means breaking through the $\bigo(n^{2.5})$ barrier. By applying
%the bound 
\eqref{eq:legallurrutia}, we finally get $x = 0.7232758$,
% and $\gamma \le 2.4465522$,  
which yields our claimed running time of $\bigo(n^{2.4465522})$.

\paragraph*{Max-min matrix product.}
We exploit fast algorithms for the max-min matrix  product. In more detail, let $A$ be an $n\times p$ matrix and $B$ be a $p\times n$ matrix. 
The entries of $A$ and $B$ are assumed to come from $\mathbb{Z} \cup \{-\infty\}$. 
The max-min product $C=A \ovee B$ is specified by $C[i,j] = \max_k \min \{A[i,k],B[k,j]\}$. 
A simple modification of the algorithm and analysis in \cite{DBLP:conf/soda/DuanP09} implies the following complexity (for which we provide a short proof in Section \ref{sec:recangular-max-min-product} for completeness).

\begin{restatable}{thmm}{thmmaxminproduct}\textsc{(Corollary of \cite{DBLP:conf/soda/DuanP09})}\label{theorem:max-min-product-bound}
%\begin{theorem}
\label{th:rectangularminmax}
If $A$ and $B$ are respectively $n \times p$ and $p \times n$ matrices, then the $A \ovee B$ product can be computed in time $\widetilde\bigo(\sqrt{\textrm{MM}(n,p,n) \cdot n^2 p}\,)$.
%\end{theorem}
\end{restatable}

\section{Fast Chain-Antichain Decomposition}
\label{sec:chainantichain}

Let $G=(V,E)$ be a DAG, with $n$ vertices and $m$ edges.
A \emph{chain} (of size $k$) of $G$ is a subset of vertices $\{c_1,\ldots,c_k\}$ such that $c_1,\ldots,c_k$ is a directed path in $G$. An \emph{antichain} (of size $k$) of $G$ is a subset of vertices $\{a_1,\ldots,a_k\}$ such that there is no edge between them. Observe that if $G$ is the graph induced by a partial order, then our definitions coincide with the usual definitions of chain and antichain in a partially ordered set.
\begin{definition}
\label{def:chain-antichain-decomposition}
An $(a,b)$-decomposition of a DAG $G$ consists of a collection  ${\cal P}$ of chains of $G$, $|{\cal P}|\leq a$, and a collection ${\cal Q}$ of antichains of $G$, $|{\cal Q}| \leq b$, that together span all the vertices of $G$.
\end{definition}

One could find an  $(\ell,n/\ell)$-decomposition with a simple greedy method, as follows:  find and remove the longest chain in $G$, for $\ell$ times in total. Next, cover whatever remains with $n/\ell$ antichains (since the remaining graph has depth at most $n/\ell$). Such an  algorithm takes time $\bigo(m \ell)$ in total ($\bigo(n^2 \ell)$ for dense graphs), since finding ``naively''  the longest chain in a DAG can be accomplished by a single graph traversal. Unfortunately, this running time would be too slow for our purposes.

We next present a faster algorithm for computing an $(\ell,\frac{2n}{\ell})$-decomposition of a DAG $G$, as stated in the following theorem, which will be proven in this section.

\begin{algorithm}[t]
\caption{Compute chain/antichain decomposition}\label{alg:decomposition}
% 	\\
\textbf{Input:} DAG $G$ with the topological ordering $v_1, \ldots, v_n$ of its vertices and parameter $\ell$, with $h=\lceil n/\ell\rceil$. \\
\textbf{Output:} $(\ell, \frac{2n}{\ell})$-decomposition of $G$.
% 	\vspace*{-5pt}
% 	\noindent\rule{\linewidth}{0.2pt}
	
	\begin{algorithmic}[1]
    \STATE Initialize $G' \gets \emptyset$, $\cal P \gets \emptyset$, $\cal Q \gets \emptyset$, $L'_i \gets \emptyset$ for $1 \leq i \leq h$  \;
	\FOR{$t = 1, \ldots, n$}
	    \STATE $\MOVE \gets \emptyset$, $\DEL \gets \emptyset$ \;
	    \STATE $\procinsert(v_t)$ \;
	    
	    \IF {$|L'_{h(v_t)}| \geq \ell$}
	        \STATE Add antichain $L'_{h(v_t)}$ to $\cal Q$ and all its vertices to~$\DEL$ \;

	    \ELSIF {$h(v_t) = h$}
	        \STATE $u \gets v_t$ \;
	        \STATE $C_t \gets \{u\}$ \;
	        \WHILE {$L_{\nextrm}(u) \neq \emptyset$}
	            \STATE $u \gets $ any element of $L_{\nextrm}(u)$\;
	            \STATE Add $u$ to $C_t$ \;
	        \ENDWHILE
	        \STATE Add chain $C_t$ to $\cal P$ and all its vertices to $\DEL$ \;
	    \ENDIF 
	    \;

	    \WHILE{$\DEL \cup \MOVE \neq \emptyset$}
	        \IF {$\DEL \neq \emptyset$}
	            \STATE Extract $w$ from $\DEL$ and execute $\procdelete(w)$\;
	        \ELSE
	            \STATE Extract $w$ from $\MOVE$ and execute $move(w)$\;
	            \IF {$|L'_{h(w)}| \geq \ell$}
	                \STATE Add antichain $L'_{h(w)}$ to $\cal Q$ and all its {vertices} to $\DEL$\;
	            \ENDIF
	        \ENDIF
	    \ENDWHILE
	\ENDFOR
	\RETURN $(\cal P, \cal Q)$
    \end{algorithmic}
\end{algorithm}

\begin{theorem}
\label{thr:mainDecomposition}
Let $G=(V,E)$ be a DAG with $n$ vertices, and let $\ell\in [1,n]$ be an integer parameter. There exists an $\bigo(n^2)$ time deterministic algorithm to compute
an $(\ell,\frac{2n}{\ell})$-decomposition of $G$.
\end{theorem}
We assume that $G$ is represented via an adjacency matrix (otherwise, we can construct it in $\bigo(n^2)$ time). The high level idea is as follows. Let $v_1,\ldots,v_n$ be a topological ordering of the vertices of $G$ (which can be computed in $\bigo(n^2)$ time). 
Let  $V_i=\{v_1,\ldots,v_i\}$, $1\leq i\leq n$, denote the first $i$ vertices in the topological order. The algorithm consists of $(n+1)$ iterations. 
At the beginning of iteration $t\geq 1$ we are given an input graph $G_{t-1}=G[W_{t-1}]$ induced in $G$ by a set of vertices $W_{t-1}\subseteq V_{t-1}$. Initially $G_0$ is the empty graph (and $W_0=\emptyset$).
For $1\leq t\leq n$, graph $G_t$ is obtained from $G_{t-1}$ as follows. We first add vertex $v_t$. Then we remove (and add to our decomposition) possibly one chain of size at least  $n/\ell$ and possibly some antichains of size $\ell$ each.
After iteration $n$, there is a final special iteration $n+1$ where $G_n$ is decomposed into at most $n/\ell$ antichains which are added to our decomposition. 
Clearly, this process produces at most $\ell$ chains and at most $2n/\ell$ antichains, as required.

As mentioned earlier, during a given iteration we insert and remove sets of vertices in a form of chains and antichains. 
We let $G'=G[W']$ denote the current graph. The set of vertices that are present in $G'$ is implicitly maintained by using a Boolean vector that indicates the existence of a vertex in $W'$. Since each vertex is added and removed from $G'$ at most once, the maintenance of this vector takes total time $\bigo(n)$. 

Furthermore, we let $L'_1,\ldots,L'_h$, $h:=\lceil n/\ell \rceil$, be disjoint (initially empty) sets of vertices that we will call \emph{layers}. We say that a vertex $v\in L'_i$ is at \emph{level} $i$. Intuitively, the level of vertex $v \in G'$ will be the length of the longest chain ending at $v$. This will immediately imply that all layers will form antichains.
During the execution of our algorithm it holds that $|L'_i| \leq \ell$, $1 \leq  i \leq h$, at all times: as a consequence,
%From a high level perspective, 
whenever at any point during the execution of the algorithm,  we identify $|L'_i| = \ell$ for some set $L'_i$, we can remove  from $W'$ the vertices in $L'_i$ since they form an antichain of size $\ell$.
We say that the algorithm is in a \emph{stable state} when the set $W'$ is partitioned into the sets $L'_{i}, 1\leq i \leq h$, such that each vertex $v\in L'_i$, for $i\geq 2$, has a parent $u\in L'_{i-1}$. 
Therefore, once we have that $L'_h \not= \emptyset$ during a stable state of the algorithm, it can be seen (as we will show later) that starting from a vertex $v\in L'_h$ and following any path by traversing a parent of each visited vertex produces (the reverse of) a chain of $G'$ of length exactly $h$.
After the removal of some set of vertices (either a chain or antichain) from $G'$, the algorithm might enter into an \emph{unstable state} (i.e., not a stable state), and hence our algorithm will work to restore a stable state by suitably  modifying the partitioning of $W'$ into the sets $L'_{1}, \dots , L'_{h}$. We next give the low level details of our algorithm.

We maintain all sets $L'_i$ into an array of size $h$ of doubly-linked lists. We will guarantee that each $v\in G'$ is contained in precisely one such set $L'_i$, and maintain bi-directional pointers between the corresponding two copies of $v$. We also maintain the sizes $|L'_i|$, and maintain the following quantities for each vertex $v\in G'$:
\begin{itemize}\itemsep0pt
\item the level $h(v)$ of $v$;
\item a list $L_{\nextrm}(v)$ of pointers to parents of $v$ in $L_{h(v)-1}$ ($L_{\nextrm}(v)=\emptyset$ for $h(v)=1$). 
%Notice that by the invariant such list has size at most $\ell$; 
\item a list $L_{\prevrm}(v)$ of pointers to children of $v$ in $L_{h(v)+1}$.
\end{itemize}

The lists $L_{\nextrm}$ and $L_{\prevrm}$ are used to assist fast insertions (resp., deletions) of vertices to (resp., from) a list $L'_i$. 
%Essentially, when a vertex $v$ is deleted from a list $L'_i, i\geq 2$, we need to update the parents of $v$ in  $L_{h(v)-1}$  (i.e., list $L_{\nextrm}(v)$) and add $v$ as the parent of its children in $L'_{h(v)+1}$ (i.e., their lists $L_{\nextrm}$). 
%This is needed so that we can quickly check whether a vertex $v$ has a parent in $L'_{h(v)-1}$ as vertices get inserted to and deleted from $L'_{h(v)+1}$.
We note that the lists $L_{\prevrm}$ are not required for the correctness of the algorithm, but only for efficiency reasons. %In particular, we use the list $L_{\prevrm}(v)$ to pinpoint the children $w$ of $v$ in $L'_{h(v)+1}$ whose lists $L_{\nextrm}(w)$ need to be updated. 
In order to be able to quickly update lists $L_{\prevrm}(v)$ and $L_{\nextrm}(v)$ after we delete or move a vertex we store together with each entry $w\in L_{\prevrm}(v)$ a pointer to the occurrence of $v$ in the list $L_{\nextrm}(w)$. 
We also store with each entry $v\in L_{\nextrm}(w)$ a pointer to the occurrence of $w$ in the list $L_{\prevrm}(v)$. 
This way, for some vertex $w\in L_{\prevrm}(v)$ we can remove $v$ from $L_{\nextrm}(w)$ in constant time, and vice versa.
For the sake of simplifying the presentation, these pointers are updated implicitly and we assume that we can execute the relevant insertions and removals in constant time.

\begin{algorithm}[t]
	\floatname{algorithm}{Procedure}
	\label{alg:insert}
	\caption{$\procinsert(v)$: Insert vertex $v$} 
	%\DontPrintSemicolon
	%\SetKw{Break}{break}
	
	%\SetKwProg{procedure}{Procedure}{}{}
	
	\begin{algorithmic}[1]
	    \STATE $L_{\prevrm}(v) \gets \emptyset$, $L_{\nextrm}(v) \gets \emptyset$ \;
	    \FOR{$i = h-1, h-2, \ldots, 1$}
	        \FOR{\textbf{each} $u \in L'_i$}
	            \IF{$(u, v) \in E(G')$}
	                \STATE Add $u$ to $L_{\nextrm}(v)$\;
	            \ENDIF
	        \ENDFOR
	        
	        \IF{$L_{\nextrm}(v) \neq \emptyset$}
	            \STATE Add $v$ to $L'_{i+1}$, increment $|L'_{i+1}|$, and set $h(v) \gets i+1$ \;
	            \FOR{$u \in L_{\nextrm}(v)$}
	                \STATE Add $v$ to $L_{\prevrm}(u)$\;
	            \ENDFOR
	            \RETURN
	        \ENDIF
	    \ENDFOR
	   \STATE Add $v$ to $L'_1$, increment $|L'_1|$, and set $h(v) \gets 1$\;	        
	\end{algorithmic}
\end{algorithm}

During each iteration $t\leq n$ we perform three main operations: 
\begin{itemize}
    \item $\procinsert(v)$: adds vertex $v$ to $G'$. This is applied once to $v_t$ at the beginning of iteration $t$.
    \item $\procdelete(v)$: deletes vertex $v$ from $G'$. This is used to remove chains and antichains from $G'$.
    \item $\procmove(v)$: moves $v$ from some $L'_i$ to some $L'_j$, $j<i$. This is used to modify the assignment of the vertices of $G'$ to the layers $L'_1,\dots, L'_{h}$ in order to restore a stable state of the algorithm.
\end{itemize}

The latter two operations can be performed multiple times in each iteration. 
%
%\pnote{The following paragraph is a repetition, so it should be deleted}
%We will preserve the following invariant: 
%As mentioned earlier, during a given iteration we insert and remove vertices. 
%We let $G'=G[W']$ denote the current graph. $G'$ is implicitly maintained by using a Boolean vector to store the current values of its vertex set $W'$. Since each vertex is added and removed from $G'$ at most once, this takes total time $\bigo(n)$. 
%Furthermore, we let $L'_1,\ldots,L'_h$, $h=\lceil n/\ell \rceil$, be disjoint (initially empty) sets of vertices that we will call \emph{layers}. If vertex $v\in L'_i$, we say that $v$ is at \emph{level} $i$. 
%
Throughout, we will maintain the following invariant: 
\begin{invariant}\label{inv:update}
After each execution of $\procinsert()$, $\procdelete()$ or $\procmove()$, the following holds:
\begin{enumerate}\itemsep0pt
\item Each vertex $v\in G'$ belongs to one $L'_i$ and, right before an $\procinsert()$ (or after the last iteration), $L'_h=\emptyset$;
\item Each $L'_i$ has size at most $\ell$, and size at most $\ell-1$ right before an $\procinsert()$ or $\procmove()$.
\item For each $v\in G'$, each parent $w\in G'$ of $v$ belongs to some lower layer $L'_j$, $j<h(v)$.
\item There is no edge $(u,v)$ for $u$ and $v$ belonging to the same layer $L'_{i}$.
\item Right before an $\procinsert()$ each vertex $v\in L'_i$, for $i\geq 2$, has a parent $u\in L'_{i-1}$ (i.e., the algorithm is at a stable state).
\end{enumerate}
\end{invariant}

\begin{algorithm}[t]
	\floatname{algorithm}{Procedure}
	\label{alg:delete}
	\caption{$\procdelete(v)$: Delete vertex $v$}
	%\DontPrintSemicolon
	%\SetKw{Break}{break}
	
	%\SetKwProg{procedure}{Procedure}{}{}
	%\procedure{\procdelete(v)}{
	\begin{algorithmic}[1]
	    \STATE Remove $v$ from $G'$ and from $L'_{h(v)}$, decrement~$|L'_{h(v)}|$ \;
	    
        \FOR{\textbf{each} $u \in L_{next}(v)$}
            \STATE Remove $v$ from $L_{\prevrm}(u)$ and $u$ from $L_{\nextrm}(v)$\;
        \ENDFOR

	    \FOR{\textbf{each} $w\in L_{\prevrm}(v)$}
	        \STATE Remove $v$ from $L_{\nextrm}(w)$ and $w$ from $L_{\prevrm}(v)$\;
	        \IF{$L_{\nextrm}(w)=\emptyset$}
	          \STATE Add $w$ to $\MOVE$\;
	        \ENDIF
	    \ENDFOR
\end{algorithmic}

\end{algorithm}

We next describe in more detail a given iteration $t\leq n$, modulo a detailed  description of the operations $\procinsert()$, $\procdelete()$ and $\procmove()$ which will be given later. We create two empty lists $\DEL$ and $\MOVE$. Intuitively, $\DEL$ contains vertices that have to be deleted from $G'$, while $\MOVE$ contains vertices that need to be moved to a lower layer (unless they are deleted earlier) in order to restore a stable state of the algorithm. 
Initially we execute $\procinsert(v_t)$. This way we add $v_t$ to $G'$  to some $L'_i$. Then we add some vertices to $\DEL$ if one of the following two cases happens: (a) $v_t\in L'_h$ or (b) $v_t\in L'_i$ and $|L'_i|=\ell$. In case (a) we compute a set $C_t$ iteratively as follows. Initially $u=v_t$. We add $u$ to $C_t$, then update $u$ to any vertex in $L_{\nextrm}(u)$ and iterate. We halt when $L_{\nextrm}(u)=\emptyset$. $C_t$ is added to the set ${\cal P}$ of chains in the decomposition under construction and its vertices are added to $\DEL$.
Notice that by Invariant \ref{inv:update} the algorithm is at a stable state right before the $\procinsert()$ operations is executed.
We will later show that, indeed, $C_t$ is a chain in $G'$ of size precisely $h$.
 In case (b) we add $L'_i$ (interpreted as a set of vertices) to the set of antichains ${\cal Q}$ in the decomposition and its vertices to $\DEL$. 
By Invariant \ref{inv:update}, there is no edge between any two vertices in $L'_i$, and thus, is an antichain in $G'$ of size precisely $\ell$.

Now we perform the following steps while $\DEL\cup \MOVE\neq \emptyset$\footnote{Notice that if cases (a) and (b) above do not happen, we stop at this point.}. 
If $\DEL\neq \emptyset$, we extract $v$ from $\DEL$ and call $\procdelete(v)$. 
This procedure removes $v$ from $G'$ and from the corresponding layer $L'_i$, and it might add some vertices to $\MOVE$. 
In particular, if $v$ used to be the only parent of $w$, then $w$ is added to $\MOVE$. 
Notice that at that point the algorithm is in an unstable state and cannot return to a stable state before all vertices in $\MOVE$ are re-assigned to appropriate layers.

Otherwise (i.e. $\DEL=\emptyset$), we extract $v$ from $\MOVE$ and, if $v\in G'$ (i.e., $v$ was not deleted in some previous step), we call $\procmove(v)$. 
This procedure will move $v$ from its current layer $L'_i$ to some lower layer $L'_j$, $j<i$. 
If after this step it happens that $|L'_j|=\ell$, then $L'_j$ is added to ${\cal Q}$ and its vertices are added to $\DEL$.
Again, by Invariant \ref{inv:update}, there is no edge between any two vertices in $L'_j$, and thus, it is an antichain in $G'$ of size $\ell$.

Procedure $\procinsert(v)$ works as follows. We consider the layers $j=h-1,\ldots,1$ in this order, and check whether $v$ has some parent in $L'_j$. 
Notice that all parents of $v$ must have been inserted at some previous iteration, however they might not belong to $G'$ any longer due to deletions. 
As soon as one such parent is found, $v$ is added to $L'_{j+1}$ (and $|L'_{j+1}|$ is incremented). 
We initialize $L_{\nextrm}(v)$ with the parents of $v$ in $L'_j$ and add $v$ to $L_{\prevrm}(u)$ for each $u\in L_{\nextrm}(v)$. 
We also set $L_{\prevrm}(v)=\emptyset$ (the children of $v$ still need to be inserted).
If no parent is found, $v$ is inserted in $L'_1$ (and $|L'_{1}|$ is incremented) and we set $L_{\nextrm}(v)=L_{\prevrm}(v)=\emptyset$. In any case $v$ is added to $G'$. 
%This procedure can be performed in time $\bigo(\ell)$ per explored level (hence $\bigo(n)$ in total). 

Procedure $\procdelete(v)$ works as follows. 
Assume $v\in L'_i$. The first step is to remove $v$ from $G'$ and $L'_i$ (decrementing $|L'_i|$). 
Then, for each vertex $u\in L_{\nextrm}(v)$ (i.e. a parent of $v$) that is still in $G'$, we remove $v$ from $L_{\prevrm}(u)$. 
Next we scan the list $L_{\prevrm}(v)$ and for each vertex $w\in G'$ in such list we remove $v$ from $L_{\nextrm}(w)$ and $w$ from $L_{\prevrm}(v)$. 
If $L_{\nextrm}(w)=\emptyset$ after the removal of $v$, we add $w$ to $\MOVE$. 

\begin{algorithm}[t]
	\floatname{algorithm}{Procedure}
	\label{alg:move}
	\caption{$\procmove(v)$: Move vertex $v$ to a lower layer}
%	\DontPrintSemicolon
%	\SetKw{Break}{break}
	
	%\SetKwProg{procedure}{Procedure}{}{}
	%\procedure{\procmove(v)}{
	\begin{algorithmic}[1]
	    \STATE Remove $v$ from $L'_{h(v)}$, decrement $|L'_{h(v)}|$\;
	    
	    \FOR{\textbf{each} $w \in L_{\prevrm}(v)$}
	        \STATE Remove $v$ from $L_{\nextrm}(w)$\;
	       \IF{$L_{\nextrm}(w) = \emptyset$}
	          \STATE Add $w$ to $\MOVE$ \;
	       \ENDIF
	   \ENDFOR

	    \STATE $L_{\nextrm}(v) \gets \emptyset$, $L_{\prevrm}(v) \gets \emptyset$ \;
	    \FOR {$j = h(v)-2, \ldots, 1$}
	        \FOR{\textbf{each} $u \in L'_j$}
	            \IF{$(u, v) \in E(G')$}
	                \STATE Add $u$ to $L_{\nextrm}(v)$ \;
	            \ENDIF
	        \ENDFOR
	        
	        \IF{$L_{\nextrm}(v) \neq \emptyset$}
	            \STATE Add $v$ to $L'_{j+1}$, increment $|L'_{j+1}|$, and set $h(v)\gets j+1$\;
	            
    	            \FOR{$u\in L_{\nextrm}(v)$}
    	                \STATE Add $v$ to $L_{\prevrm}(u)$ \;
    	            \ENDFOR
	            
	            \RETURN   
	        \ENDIF
	    \ENDFOR
	    \STATE Add $v$ to $L'_1$, increment $|L'_1|$, and set $h(v)\gets 1$ \;
	    
	    \FOR{\textbf{each} $u \in L'_{2}$}
            \IF{$(v, u) \in E(G')$}
                \STATE Add $u$ to $L_{\prevrm}(v)$ and $v$ to $L_{\nextrm}(u)$\;
            \ENDIF
        \ENDFOR
	\end{algorithmic}
\end{algorithm}

It remains to describe $\procmove(v)$. Again assume $v\in L'_i$. Notice that by construction $i\geq 2$ since we never add to $\MOVE$ vertices in $L'_1$. 
%Observe that at this point there are no vertices to be deleted. 
We initially consider the vertices $w\in L_{\prevrm}(v)$, and remove $v$ from $L_{\nextrm}(w)$ and $w$ from $L_{\prevrm}(v)$. 
Notice that, similarly to the $\procdelete(v)$ case, if $L_{\nextrm}(w)=\emptyset$, we need to add $w$ to $\MOVE$.
Then we consider the layers $j=i-2,\ldots,1$ one by one, and check whether $L'_j$ contains at least one parent of $v$. 
If such a parent is found, $v$ is moved from $L'_i$ to $L'_{j+1}$ (updating $|L'_i|$ and $|L'_{j+1}|$, and setting $h(v)\gets j+1$, consequently).
All the parents of $v$ in $L'_j$ are added to $L_{\nextrm}(v)$. 
If no parent is found, $v$ is moved to $L'_1$ and the procedure sets $L_{\nextrm}(v)=\emptyset$, and $h(v)\gets 1$. 
In either case (that is, either $L_{\nextrm}(v) =\emptyset$ or $L_{\nextrm}(v)\neq \emptyset$), we scan all vertices of $L'_{h(v)+1}$ for children of $v$ and for each such child $u$ we add $u$ to $L_{\prevrm}(v)$ and $v$ to $L_{\nextrm}(u)$. 

Notice that the above procedure moves a vertex only to a strictly lower level.

Observe that, by giving priority to the $\procdelete()$ operations over the $\procmove()$ operations, we avoid increasing the size of any layer above $\ell$ (that is, we preserve case 4 of Invariant \ref{inv:update}). This not only allows us to identify antichains of length $\ell$ during an unstable state but also, most importantly, limits the size of the vertices to test for identifying parents and children during the subsequent $\procinsert()$ and $\procmove()$ operations.
\emph
{We defer to the Section~\ref{seclaterproofs} the proof that after each execution of $\procinsert()$, $\procdelete()$ or $\procmove()$ Invariant \ref{inv:update} is satisfied.}

At the end of the last iteration by Invariant \ref{inv:update} all vertices still in $G'$  are contained in some $L'_i$, $i=1,\ldots,h-1$. We execute a special final iteration $n+1$ where we add each such set $L'_i$ as an antichain to our decomposition.  

\begin{lemma}\label{lem:correctnessDecomposition}
The  algorithm described above (pseudo-code in Algorithm \ref{alg:decomposition}) computes an $(\ell,\frac{2n}{\ell})$-decomposition.
\end{lemma}
\begin{proof}
By construction each vertex which is included in a chain or antichain in the first $n$ iterations is deleted from $G'$, hence it is not included in any following chain or antichain. The antichains added to ${\cal Q}$ in iteration $n+1$ are disjoint by Invariant \ref{inv:update}. Furthermore, all vertices are added at some point to $G'$, hence they are included by construction in some chain or antichain at some later point. Thus the chains and antichains induce a partition of the vertex set $V$.

Each list $L_{next}(v)$ by construction contains parents of $v$ only. Furthermore, right before the $\procinsert(v_t)$ operation that leads to the construction of some set $C_t$, each vertex $w\in L'_i$, $i\geq 2$, must have at least one parent in $L'_{i-1}$ by construction (by Invariant \ref{inv:update}), hence $L_{\nextrm}(w)$ is not  empty and contains vertices in $G'$. Consequently $C_t$ is a chain in $G'$ of size precisely $h$. 

Similarly, by Invariant \ref{inv:update}, vertices in each set $L'_i$ that are added to the set ${\cal Q}$ of antichains are not parents of each other, hence they form a correct antichain. Notice also that all the sets $L'_i$ that are added to ${\cal Q}$ in the first $n$ iterations have size precisely $\ell$ and consist of vertices in $G'$ only. Indeed, the condition $|L'_i|=\ell$  happens after an $\procinsert()$ or $\procmove()$ operation. In both cases there are no vertices in $\DEL$, hence any vertex in $L'_i$ is also present in $G'$.  

Finally, we bound the number of chains and antichains. As argued before, each chain $C_t$ has size precisely $h \geq n/\ell$. Hence disjointness implies that there are at most $\ell$ such chains. Similarly, each antichain that we add in the first $n$ iterations has size precisely $\ell$, hence disjointness implies that there are at most $n/\ell$ such antichains. In the final iteration $n+1$ we add at most $h-1 \leq n/\ell$ extra antichains. The claim follows. 
\end{proof}

\emph{We defer to the Section~\ref{seclaterproofs} the proof that the running time of the algorithm is $\bigo(n^2)$.}

%The proof of Theorem \ref{thr:mainDecomposition} follows trivially from Lemmas \ref{lem:correctnessDecomposition} and \ref{lem:timeDecomposition}.

\section{All-Pairs LCA in DAGs}
\label{sec:mainresult}

In this section we present our improved algorithm for All-Pairs LCA in DAGs.

We start by sketching the high level ideas behind the algorithm. Let $G_{input}$ be the input DAG and let $G$ be the transitive closure of $G_{input}$. We compute $G$ in $\bigo(n^\omega)$ time and  solve the All-Pairs LCA problem on $G$ (obviously the solution in the two cases is identical). 

To do this,  we first compute an $(n^x,2n^{1-x})$-decomposition $(\cP,\cQ)$ of $G$ in $\bigo(n^2)$ time with the algorithm from Theorem \ref{thr:mainDecomposition}. Recall that $\mathcal{P}=\{P_1, \dots, P_p\}$ is a set of $p\leq n^x$ chains and $\mathcal{Q}=\{Q_1, \dots, Q_q\}$ a set of $q\leq 2n^{1-x}$ antichains. Here $x\in [0,1]$ is a parameter to be optimized later in order to minimize the overall running time. 

We now define %in a natural way 
the notion of LCA restricted to a subset $W$ of vertices as follows. 
\begin{definition}
Given a DAG $G=(V,E)$, a subset of vertices $W\subseteq V$, and a  pair of vertices $u,v\in V$, $LCA_W(u,v)$ is the set of vertices  $w\in W$ which are ancestors of both $u$ and $v$ and such that there is no descendent $w'\in W$ of $w$ with the same property. Any $w\in LCA_W(u,v)$ is a $W$-restricted LCA of $\{u,v\}$. The $W$-restricted All-Pairs LCA problem is to compute $\lca_W(u,v)\in LCA_W(u,v)$ for all pairs of vertices $u,v\in V$ ($\lca_W(u,v)=-\infty$ if $LCA_W(u,v)=\emptyset$).   
\end{definition}
We use $\cP$-restricted and $\cQ$-restricted as shortcuts for $(\cup_{P\in \cP}P)$-restricted and $(\cup_{Q\in \cQ}Q)$-restricted resp., and also define analogously $\textrm{LCA}_{\cP}(\cdot,\cdot)$, $\lca_{\cP}(\cdot,\cdot)$ etc. The next step is to solve the $\cP$-restricted and $\cQ$-restricted All-Pairs LCA problems. In particular, we plan to compute the values $\lca_{\cP}(u,v)$ and $\lca_{\cQ}(u,v)$ for all pairs of vertices $u,v\in V$. This is explained in sections \ref{sec:Prestricted} and \ref{sec:Qrestricted} resp. In more detail, the first problem is solved in time $\widetilde{\bigo}(n^{\frac{\omega(1,x,1) + 2 + x}{2}})$ using a reduction to one max-min product. The second problem is solved in time $\widetilde{\bigo}(n^{1-x+\omega(1,x,1)})$ by performing one Boolean matrix product of cost $\widetilde{\bigo}(n^{\omega(1,x,1)})$ for each $Q\in \cQ$.

At this point we need to combine the two solutions together. A naive approach might be as follows. Let us label the vertices from $1$ to $n$ according to some arbitrary topological order. Then, for any pair of vertices $u,v$, we simply set $\lca(u,v)=\max\{\lca_{\cP}(u,v),\lca_{\cQ}(u,v)\}$ (in total time $\bigo(n^2)$). Unfortunately, as discussed in Section \ref{sec:global}, there exist  topological orderings for which this approach fails. In the same section we show how to compute a specific topological ordering in $\bigo(n^2)$ time such that the above combination indeed works. Then, it will be sufficient to optimize over the parameter $x$.

Throughout this section we assume that vertices are labeled with integers between $1$ and $n$ (according to some given order to be specified later).

\subsection{Computing $\cP$-Restricted LCAs}
\label{sec:Prestricted}

In this section we present our algorithm for the $\cP$-restricted All-Pairs LCA problem. We next assume that vertices are labeled with integers $1,\ldots,n$ according to some topological order. In the next section we will specify such ordering in a more careful way in order to achieve our final result.

Our algorithm works as follows (see also the pseudo-code in Algorithm \ref{alg:chain-restricted-lca}). For each vertex $v$ and each chain $P_i$, we compute the ancestor $w_i(v)$ of $v$ in $P_i$ with largest index ($w_i(v)=-\infty$ if there is no such ancestor). Next, for each pair of vertices $u,v$ and each $P_i$, we compute $w_i(u,v)=\min\{w_i(u),w_i(v)\}$. Finally we set $\lca_{\cP}(u,v)=\max_{1 \le i \le p} \{w_i(u,v)\}$.

Recall that, given two matrices $A$ and $B$, their max-min product $C=A \ovee B$ is specified by $C[i,j] = \max_k \min \{A[i,k],B[k,j]\}$.

In order to implement the above algorithm, it is sufficient to construct an  $n\times n^x$ matrix $A$ whose rows are indexed by vertices in $V$ and whose columns are indexes by chains $P_i$. The entry $A[v,P_i]$ corresponds to the value $w_i(v)$ defined above. Then it is sufficient to compute $C=A \ovee A^T$ and set $\lca_{\cP}(u,v)=C[u,v]$ for all pairs $u,v\in V$.

\begin{algorithm}[t]
%\SetAlgoLined
\textbf{Input:} Transitive closure graph $G=(V,E)$, 
and a family of chains $\mathcal{P}=\{P_1, \dots, P_p\}$ of $G$ where $p\leq n^{x}$.\\
\textbf{Output:} {$\mathcal{P}$-restricted LCA $\lca_{\mathcal{P}}(u,v)$ for each pair of vertices $u,v\in V$.}

    \begin{algorithmic}[1]
	    \STATE Initialize $\lca_{\cP}(\cdot,\cdot)$ with $-\infty$\;
	    
	    \STATE Let $A$ be an $n \times p$ matrix\;
	    
	    \FOR{$P_i \in \mathcal{P}$}
	        \FOR{$v\in V$}
	            \STATE Let $w_i(v)$ be the parent of $v$ in $P_i$ with the largest index, otherwise $w_i(v)=-\infty$\;
	            
	            \STATE Set $A[v, i] \leftarrow w_i(v)$\;
	        \ENDFOR
	    \ENDFOR
	    
	    \STATE Compute the (max,min)-product $A \ovee A^T$\;

	    \FOR{all $u,v \in V, u\not=v$}
	            \STATE $\lca_{\mathcal{P}}(u,v) \leftarrow (A \ovee A^T)[u,v]$\;
	    \ENDFOR
	
	\end{algorithmic}
 \caption{Compute $\lca_{\mathcal{P}}(u,v)$ for all pairs of vertices $u,v\in V$.}
 \label{alg:chain-restricted-lca}
\end{algorithm}

\begin{lemma}
The $\cP$-restricted All-Pairs LCA problem can be solved in time $\widetilde{\bigo}(n^{\frac{\omega(1,x,1) + 2 + x}{2}})$.
\label{lemma:chain-restricted-lca}
\end{lemma}

\begin{proof}
Consider the above algorithm (pseudo-code in Algorithm~\ref{alg:chain-restricted-lca}). 
To analyze its running time, we observe that the matrix $A$ can be built in time $\bigo(n^2)$  
by scanning the vertices $v\in V$ and the vertices $w$ in $\cP$. The rest of the computation takes time $\bigo(n^{\frac{\omega(1,x,1) + 2 + x}{2}})$ by Theorem \ref{th:rectangularminmax}. The claim follows.

For the correctness observe that, if $P_i$ contains a vertex in $LCA_{\cP}(u,v)$, then this vertex has to be $w=w_i(u,v)$.
Indeed, by construction $w$ is an ancestor of both $u$ and $v$.
Since $w_i(u,v)=\min\{w_i(u),w_i(v)\}$, any successor $w'$ of $w$ along $P_i$ is not an ancestor of $u$ or of $v$.
Vice versa, any ancestor $w'$ of $w$ along $P_i$ cannot be in $LCA_{\cP}(u,v)$ due to the existence of $w$.
Therefore the set $W:=\{w_i(u,v)\}_i$ contains $LCA_{\cP}(u,v)$.
Notice also that $W=\{-\infty\}$ iff $u$ and $v$ do not have a common ancestor in $\cP$, in which case $LCA_{\cP}(u,v)=\emptyset$.
Therefore we can w.l.o.g. assume that the algorithm returns some $w\in W$, $w\neq -\infty$.
In particular, $w$ is a vertex with the largest index in $W$ according to the considered topological order.
Assume by contradiction that $w\notin LCA_{\cP}(u,v)$.
This implies that there exists some other vertex $w'\in LCA_{\cP}(u,v)$ which is a descendant of $w$.
But vertex $w'$ must be contained in $W$, which implies $w'<w$ (otherwise the algorithm would not return $w$).
This is a contradiction since $w'$ is a descendant of $w$ and at the same time has a smaller index in some topological order.
\end{proof}

\subsection{Computing $\cQ$-Restricted LCAs}
\label{sec:Qrestricted}

In this section we present our algorithm for the $\cQ$-restricted All-Pairs LCA problem.
For notational convenience let us rename $\cQ$ as $\cQ'=\{Q'_1, \dots, Q_{q'}\}$. Recall that $q'\leq 2n^{1-x}$. 
The first step in our construction is to transform $\cQ'$ into a more convenient family of antichains $\cQ$ as follows. 
\begin{definition}
Let $\cQ=\{Q_1, \dots, Q_q\}$ be a collection of disjoint antichains of a transitive closure graph $G=(V,E)$. $\mathcal{Q}$ is \emph{path-respecting} if for any two vertices $x\in Q_i,y \in Q_j$ such that $(x,y) \in E$ it holds that $i<j$.  
\end{definition}
\begin{lemma}[Folklore]\label{lemma:total-order-antichains}
Given a transitive closure graph $G=(V,E)$ and a collection of $q'$ disjoint antichains $\cQ'=\{Q'_1, \dots, Q'_{q'}\}$ over the vertex set $W\subseteq V$, a greedy algorithm computes a partition of $W$ into a collection of $q\leq q'$ disjoint antichains $\cQ=\{Q_1, \dots, Q_{q}\}$ in time $\bigo(n^2)$.
\end{lemma}
\begin{proof}
Let us initialize $W'$ to $W$. 
The greedy algorithm proceeds in rounds. In round $i$ we set $Q_i = \{$all vertices with indegree $0$ in $G[W']\}$. 
Then $Q_i$ is added to $\cQ$, and its vertices are removed from $W'$. 
We halt when $W'=\emptyset$. 
It is easy to see that each $Q_i$ is indeed an antichain. 
Let $q$ be the number of antichains produced by the algorithm and let $h$ be the height of $G[W]$, that is the size of its longest chain. 
We have that $h \le q'$, since by Mirsky's theorem (c.f. \cite{mirsky1971dual}) size of any antichain cover of $G[W]$ is at least $h$. 
We also have $h = q$, since greedy algorithm reduces the length of longest chain in $G[W']$ by exactly one at each iteration.

The above algorithm can be easily implemented in time $\bigo(n^2)$. 
Indeed, it is sufficient to maintain the in-degree of the vertices and update them each time a vertex is removed. 
Whenever during an iteration the in-degree of some vertex $v$ becomes $0$ because of the removal of other vertices, we add $v$ to a list of vertices to be used in the next round.
\end{proof}

We use Lemma \ref{lemma:total-order-antichains} to transform $\mathcal{Q}'$ into a path-respecting family of $q\leq q'\leq 2n^{1-x}$ antichains $\mathcal{Q}=\{Q_1, \dots, Q_{q}\}$.
It remains to solve the $\cQ$-restricted All-Pairs LCA problem. To this aim, we use a relatively simple reduction to Fast Boolean Matrix Multiplication. Let $C=A\cdot B$ be the product of an $n\times p$ Boolean (i.e., $0$-$1$) matrix $A$ and a $p\times n$ Boolean matrix $B$. The \emph{witness matrix} $W$ of this product is an $n\times n$ matrix where $W[i,j]$ is any index $k$ such that $A[i,k]=B[k,j]=1$. We conventionally set $W[i,j]=-\infty$ if no such index exists. Recall that the time needed to compute $C$ is denoted by $\textrm{MM}(n,p,n)$. A mild adaptation of the algorithm and analysis in \cite{DBLP:conf/focs/AlonGMN92} shows that we can compute $W$ roughly in the same amount of time.   

\begin{theorem}[Folklore, corollary of \cite{DBLP:conf/focs/AlonGMN92}]
\label{th:rectwitness}
The witness matrix $W$ of the product $C=A\cdot B$ of an $n\times p$ Boolean matrix $A$ and a $p\times n$ Boolean matrix $B$ can be computed in time $\widetilde{\bigo}(\textrm{MM}(n,p,n))$ by a deterministic algorithm.
\end{theorem}

Our algorithm works as follows (see also the pseudo-code in Algorithm \ref{alg:antichain-restricted-lca}). We initialize $\lca_{\cQ}(\cdot,\cdot)$ with $-\infty$. Then we consider the antichains $Q_q,\ldots,Q_1$ in this order. For each $Q_i$ and each pair of vertices $u,v$ with $\lca_{\cQ}(u,v)=-\infty$, we check if $Q_i$ contains a common ancestor $w$ of $v$ and $u$, in which case we set  $\lca_{\cQ}(u,v)=w$. In order to perform efficiently this step we build an $n\times n^x$ matrix $A$ whose rows are indexed by vertices  in $V$ and whose columns are indexed by vertices in $Q_i$. We set entry $A[v,w]$ to $1$ if $w$ is an ancestor of $v$ and to $0$ otherwise\footnote{Padding with zeros the columns not corresponding to vertices in $Q_i$.}. We compute the product $A \cdot A^T$ and  its witness matrix $W$. Notice that the pair $u,v$ has a common ancestor $w$ in $Q_i$ iff $A\cdot A^T[u,v]\neq 0$, in which case $W[u,v]$ contains one such vertex. Thus it is sufficient to set $\lca_{\cQ}(u,v)=W[u,v]$.

\begin{algorithm}[t]
%\SetAlgoLined
\textbf{Input:} {Transitive closure graph $G=(V,E)$, and a family of antichains $\mathcal{Q}=\{Q_1, \dots, Q_q\}$ of $G$ that is path-respecting such that $q\leq 2n^{1-x}$.}\\
\textbf{Output:} {$\mathcal{Q}$-restricted LCA $\lca_{\mathcal{Q}}(u,v)$ for each pair of vertices $u,v\in V$.}
Initialize $\lca_{\cQ}(\cdot,\cdot)$ with $-\infty$.

\begin{algorithmic}[1]
  \FOR{$i=q,\ldots,1$}
    \STATE Initialize an $n \times n^x$ matrix $A$ with zeros\;
    
 \STATE Let $\phi_{i}:Q_{i}\xrightarrow{1:1} \{1,\dots, |Q_{i}|\}$ be an arbitrary bijection and $\phi^{-1}_{i}(\cdot)$ be its inverse function\;
    
    \FOR{all $x\in V,y\in Q_{i}$ such that $(x,y)\in E$}
        \STATE $A[x, \phi_{i}(y)] \leftarrow 1$\;
    \ENDFOR
    
    \STATE Compute $A\cdot A^T$, and its witness matrix $W$\;
    
    \FOR{all $u,v \in V, u\not=v$}
        \IF{$\lca_{\mathcal{Q}}(u,v) = -\infty$ and $A\cdot A^T[u,v] \not=0$}
            \STATE $\lca_{\mathcal{Q}}(u,v) \leftarrow \phi^{-1}_{i}(W[u,v])$\;
        \ENDIF
    \ENDFOR
    
 \ENDFOR
\end{algorithmic}
 \caption{Compute $\lca_{\mathcal{Q}}(u,v)$ for all pairs of vertices $u,v\in V$.} 
 \label{alg:antichain-restricted-lca}
\end{algorithm}
\begin{lemma}\label{lemma:antichain-restricted-lca}
The $\mathcal{Q}$-restricted All-Pairs LCA problem can be solved in time $\widetilde{\bigo}(n^{1-x+\omega(1,x,1)})$.
\end{lemma}
\begin{proof}
Consider the above algorithm (pseudo-code in Algorithm \ref{alg:antichain-restricted-lca}). Its running time is upper bounded by $\widetilde\bigo(\sum_{i=1}^q \textrm{MM}(n,|Q_i|,n))$. 
Assume w.l.o.g. that $|Q_i|$ is non-increasing, then $|Q_i| \le n/i$, and by monotonicity of  $\textrm{MM}(n,\cdot,n)$
\begin{align*}
\sum_{i=1}^q \textrm{MM}(n,|Q_i|,n) &\le \sum_{i=1}^q \textrm{MM}(n,n/i,n)\\
&\le \sum_{j=0}^{\log q} 2^j \textrm{MM}(n,n/2^j,n)\\
&\le (1+\log q) \cdot q \cdot \textrm{MM}(n,n/q,n)\\
&\in \widetilde\bigo(n^{1-x+\omega(1,x,1)}).
\end{align*}

For the correctness, assume by contradiction that for some pair of vertices $u,v$ the computed value $\lca_\cQ(u,v)$ is not correct. 
Notice that $\lca_\cQ(u,v)=-\infty$ iff $u$ and $v$ have no common ancestor in $\cQ$, hence we can assume w.l.o.g. $\lca_\cQ(u,v)=w$ for some $w$ in some $Q_i$. 
The contradiction implies that there exists a common ancestor $w'\in Q_j$ of $u$ and $v$ which is a descendant of $w$ (in particular, $(w,w')\in E$ since $G$ is a transitive closure). 
Notice that $j\neq i$ since $Q_i$ is an anti-chain. 
By construction $u$ and $v$ do not have any common ancestor in $Q_{i+1},\ldots,Q_q$ since otherwise at the time when $Q_i$ is considered we would have  $\lca_\cQ(u,v)\neq -\infty$. 
Hence it must be the case that $j<i$. 
This is a contradiction since the existence of the pair $w,w'$ shows that $\cQ$ is not path-respecting.
\end{proof}

\subsection{Patching the LCAs Together}
\label{sec:global}

Suppose we are given values $\lca_\cP(\cdot,\cdot)$ and $\lca_\cQ(\cdot,\cdot)$ as computed in previous sections. Let us also assume that vertices are labeled from $1$ to $n$ according to an \emph{arbitrary} topological ordering. The following approach to solve All-Pairs LCA might be tempting: for each pair $u,v\in V$, we simply set $\lca(u,v)=\max\{\lca_\cP(u,v),\lca_\cQ(u,v)\}$. Unfortunately this approach does not work, as illustrated in Figure \ref{fig:bad-topological-rder-example}.
Intuitively, the issue is that in the computation of $\lca_\cQ(u,v)$ the algorithm can return any vertex $w$ in some $Q_i$ which is a common ancestor of $u$ and $v$, not necessarily the one with largest index in $Q_i$. This flexibility is essential to achieve the claimed running time: computing $w$ with the largest index in $Q_i$ would require a \emph{max-witness} computation, and the best-known algorithms for the latter problem are substantially slower than Boolean matrix multiplication.

In order to circumvent this problem, we will compute (in $\bigo(n^2)$ time) a more structured topological order. Using this particular order rather than an arbitrary topological order, guarantees that the above approach works.
In particular, our goal is to define a topological order such that, if $\lca_{\mathcal{P}}(u,v)$ appears later than $\lca_{\mathcal{Q}}(u,v)$ in this order, then there is no path from $\lca_{\mathcal{P}}(u,v)$ to any $\mathcal{Q}$-restricted LCA for $u,v$, and vice~versa.

\begin{definition}
Let $G=(V,E)$ be a transitive closure graph and $\mathcal{Q}=\{Q_1, \dots, Q_q\}$ a path-respecting family of antichains of $G$. A $\mathcal{Q}$-compact topological order of the vertices in $V$ is a topological order such that all vertices in an antichain $Q_i\in \mathcal{Q}$ appear consecutively and a vertex in $Q_i$ appears earlier than a vertex in $Q_j$, for $i<j$. \end{definition}
\begin{lemma}
\label{lem:respecting}
Given a transitive closure graph $G=(V,E)$ and a path-respecting family of antichains $\mathcal{Q}$, we can compute a $\mathcal{Q}$-compact topological order of $G$ in time~$\bigo(n^2)$.
\label{lemma:q-compact-topological-order-computation}
\end{lemma}
\begin{proof}
For notational convenience, let us define a dummy set $Q_0=\emptyset$. 
The algorithm proceeds in rounds. At the beginning of round $i\geq 0$ we are given a current subset of vertices $W$ and a partial topological ordering $R$ (implemented as list) of the remaining vertices $V\setminus W$. Initially $W=V$ and $R$ is empty. During round $i$ we append the vertices of $Q_i$ to $R$ in any order and remove them from $W$. Then we iteratively identify the sources $S_i$ in $G[W\setminus \left(\bigcup_{i < j \leq q} Q_j\right)]$, append the vertices $S_i$ to $R$ in any order, and finally remove them from $W$.

In order to implement the above algorithm in $\bigo(n^2)$ time, we can use an approach similar to the proof of Lemma \ref{lemma:total-order-antichains}, where we keep track of vertices whose in-degree becomes zero after the removal of other vertices.  

For the correctness, trivially by construction the indexes of the vertices in the same anti-chain $Q_i$ are consecutive, and the indexes in $Q_i$ are smaller than the indexes in $Q_j$ for $j>i$. Hence it remains to show that $R$ defines a topological order at the end of the algorithm. Suppose by contradiction that there exists an edge $(x,y)\in E$ such that $x$ is placed after $y$ in $R$. Let $y\in S_i\cup Q_i$ for some $i$. Assume by contradiction that $y\in S_i$. Then it must be the case that $x\in Q_i$ or $x$ was removed in some earlier iteration. Indeed otherwise $y$ would not be a source. In both cases $x$ would appear earlier than $y$ in $R$. It therefore remains to consider the case $y\in Q_i$, $i\geq 1$. Assume that $x\in Q_j$ for some $j$. The fact that $\cQ$ is path-respecting implies that $j<i$. This means that $x$ is added to $R$ in some earlier iteration, a contradiction. So the remaining case is that $x\in S_j$ for some $j\geq i$. Notice that vertices in $S_j$ become sources right after the removal of vertices in $Q_j$ (otherwise they would be sources at some earlier round). In particular, there must exist some parent $w$ of $x$ in $Q_j$. Notice that $w\in Q_j$ is an ancestor of $y\in Q_i$ (hence $(w,y)\in E$), and $j\geq i$. This contradicts the fact that $\cQ$ is path-respecting. 
\end{proof}

\begin{algorithm}[t]
%\SetAlgoLined
\textbf{Input:} {DAG $G_{input}=(V,E_{input})$}\\
\textbf{Output:} {$\lca(u,v)$ for each pair of vertices $u,v\in V$}

\begin{algorithmic}[1]
    \STATE Compute the transitive closure graph $G=(V,E)$ of $G_{input}$\;

    \STATE Use Algorithm \ref{alg:decomposition}
    to compute a $(n^x, 2n^{1-x})$-decomposition into a family of chains $\mathcal{P}=\{P_1, \dots, P_p\}$ with $p \leq n^x$  and a family of antichains $\mathcal{Q'} = \{Q'_1, \dots, Q'_{q'}\}$ with $q' \leq 2n^{1-x}$.\;
    
    \STATE Use Lemma \ref{lemma:total-order-antichains} with input $\mathcal{Q'}$ to compute a path-respecting family of antichains $\mathcal{Q}=\{Q_1,\dots, Q_{q}\}$ of $G$ where $q \leq 2n^{1-x}$.\;

    \STATE Compute a $\mathcal{Q}$-compact topological order of $G$ using Lemma \ref{lemma:q-compact-topological-order-computation} and rename vertices so that they are $1,\ldots,n$ according to this order\;

    \STATE Use Algorithm \ref{alg:chain-restricted-lca} to compute $\mathcal{P}$-restricted LCA $\lca_{\mathcal{P}}(u,v)$ for each pair of vertices $u,v\in V$\;

    \STATE Use Algorithm \ref{alg:antichain-restricted-lca} to compute $\mathcal{Q}$-restricted LCA $\lca_{\mathcal{Q}}(u,v)$ for each pair of vertices $u,v\in V$\;
 
    \FOR{\textbf{all} $u,v \in V, u\not=v$}
        \STATE $\lca(u,v) \leftarrow \max\{\lca_{\mathcal{Q}}(u,v), \lca_{\mathcal{P}}(u,v)\}$\;
    \ENDFOR
\end{algorithmic}
 \caption{Compute $\lca_{\mathcal{V}}(u,v)$ for all pairs of vertices $u,v\in V$.}
 \label{alg:All-Pairs-global-lca}
\end{algorithm}

This concludes the description of our algorithm for the All-Pairs LCA problem in DAGs (see also the pseudo-code in Algorithm \ref{alg:All-Pairs-global-lca}).

\begin{theorem}[Main Theorem]
\label{th:runtimebyx}
All-Pairs LCA in DAGs can be solved in time $\widetilde\bigo( n^{\gamma})$, where $\gamma = 1 + 2x$ and $x$ is the solution of the equation $3x=\omega(1,x,1)$. 
\end{theorem}
\begin{proof}
Consider the above All-Pairs LCA algorithm (for pseudo-code see Algorithm \ref{alg:All-Pairs-global-lca}). The running time of the algorithm is  $\widetilde\bigo(n^{\omega} + n^{\frac{\omega(1,x,1)+2+x}{2}}+n^{1-x + \omega(1,x,1)})$ for a fixed $x\in [0,1]$. The claimed running time is obtained by imposing $\frac{\omega(1,x,1)+2+x}{2}=1-x + \omega(1,x,1)$, and observing that  $\omega \le \omega(1,x,1)+1-x$ for any $x \ge 0$.

For the correctness assume by contradiction that $w=\lca(u,v)$ is not a correct answer.
Notice that if $u$ and $v$ have no common ancestor, by construction $w=-\infty$ and the answer is correct. 
So we can assume that $w$ is an index of some vertex. 
Assume first $w=\lca_{\cP}(u,v)$. 
By contradiction assume that $w'$ is some descendant of $w$ which is also a common ancestor of $u$ and $v$. 
Notice that $w'>w$ since we consider a topological order. 
The correctness of the $\cP$-restricted All-Pairs LCA algorithm implies that $w'$ is contained in $\cQ$. 
In particular $w'\in Q_i$ for some $i$. 
Since the considered topological order is $\cQ$-compact, all vertices in $Q_i$ appear after $w$ in the topological order (in particular, they have larger indices than $w$). 
Since $Q_i$ contains at least one common ancestor of $u$ and $v$ (namely, $w'$), by construction $\lca_{\cQ}(u,v)$ is contained in $Q_j$ for some $j\geq i$. 
Since the topological order is $\cQ$-compact, this implies $\lca_{\cQ}(u,v)> w$. 
Hence we get a contradiction $w = \lca_{\cP}(u,v)\geq \lca_{\cQ}(u,v)>w$. 

The case that $w=\lca_{\cQ}(u,v)$ is symmetric. 
In particular, any descendant $w'$ of $w$ which is a common ancestor of $u$ and $v$ must be contained in $\cP$, and $w'>w$. 
By construction we have $\lca_{\cP}(u,v) \geq w'$. 
Hence we get a contradiction $w = \lca_{\cQ}(u,v) \geq \lca_{\cP}(u,v)\geq w' > w$.
\end{proof}

\begin{figure}
    \centering
    \includegraphics[width=0.33\textwidth]{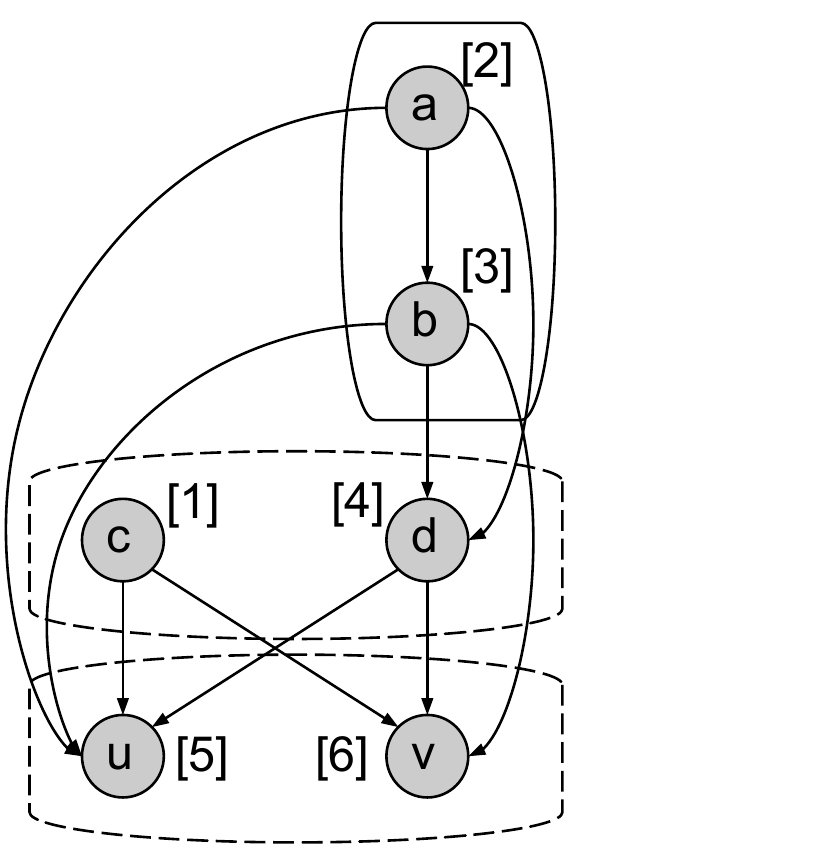}
    \caption{An example of a topological order in a transitive closure graph $G$ that is not suitable for combining the $\mathcal{Q}$-restricted LCA and the $\mathcal{P}$-restricted LCA for the pair of vertices $u,v$. The family of chains in $G$ is $\mathcal{P}=\{\{a,b\}\}$ and the path-respecting family of antichains is $\mathcal{Q}=\{\{c,d\},\{u,v\}\}$. The index of each vertex in the topological order is in brackets next to each vertex. It might happen that $\lca_\cP(u,v)=b$ and $\lca_\cQ(u,v)=c$, in which case the algorithm would return the incorrect answer $\lca(u,v)=b$. }
    \label{fig:bad-topological-rder-example}
\end{figure}

\section{Missing proofs from Section \ref{sec:chainantichain}.}
\label{seclaterproofs}
\begin{lemma}\label{lem:invariant}
After each execution of $\procinsert()$, $\procdelete()$ or $\procmove()$, Invariant \ref{inv:update} is satisfied.
\end{lemma}
\begin{proof}
We prove the claim by induction on the number of operations. In particular we will assume that the considered operation is the $k$-th one, and the invariant holds before its execution. Notice that the invariant is trivially satisfied before the first execution of any such operation (when $G'$ and the layers $L'_i$ are empty).

\smallskip\noindent(1) Consider the first part of the claim. Clearly if the $k$-th operation is $\procdelete(v)$ or $\procmove(v)$, the claim holds. In case of $\procinsert(v)$, the inductive hypothesis guarantees that all parents of $v$ in $G'$ are in level $h-1$ or lower. Hence $v$ is inserted in layer $h$ or lower. For the second part of the claim, assume inductively that $L'_h$ is empty before the execution of some $\procinsert()$ (this is true at the beginning). Observe that the only operation that can add some vertex $w$ to $L'_h$ is $\procinsert(w)$. Vertex $w$ is deleted right after the $\procinsert(w)$ operation, since we always test whether $L'_{h}\neq \emptyset$ and, if so, the algorithm retrieves and removes (the reverse of) a path that is traversed starting from $w$ and following a parent of each visited vertex (this path always includes $w$). Furthermore $\procdelete()$ and $\procmove()$ never add vertices to $L'_h$. Hence before the execution of the next $\procinsert()$ the set $L'_h$ is empty as required.     

\smallskip\noindent (2) Clearly if the $k$-th operation is $\procdelete(v)$ the claim holds. If the $k$-th operation is $\procmove(v)$ or $\procinsert(v)$ by inductive hypothesis at most one layer $L'_i$ can reach size $\ell$, while all other layers have the same or smaller size after the operation. Notice that we give priority to the $\procdelete()$ operations over the $\procmove()$ operations, and hence once $|L'_{i}|=\ell$ and all vertices of $L'_i$ are inserted into $\DEL$, no further points are inserted into $L'_i$ until it is fully empty. Thus, all vertices of $L'_i$ are deleted before the next execution of a $\procmove()$ or $\procinsert()$. The claim then holds.  

\smallskip\noindent (3) This is the most delicate claim. The claim trivially holds if the $k$-th operation is $\procdelete(v)$, and it holds by construction if it is $\procinsert(v)$. Next assume that the $k$-th operation is $\procmove(v)$, and let $v\in L'_i$ at the time of its execution. By inductive hypothesis all parents of $v$ are of level at most $i-1$ at that time, and the procedure considers all the parents of $v$ in $G'$ of level at most $i-2$. Hence assume by contradiction that $v$ has some parent in $G'$ of level $i-1$ when $\procmove(v)$ is executed (in which case the invariant is violated). Suppose that the operation that added $v$ to $\MOVE$ is the $k'$-th one, $k'<k$. Observe that this operation is either a $\procmove(w)$ or a $\procdelete(w)$ for some $w\in L'_{i-1}$. Since by assumption $v\in G'$ at the time of execution of $\procmove(v)$, by construction the level of $v$ remains $i$ during all the intermediate operations $k'+1,\ldots,k-1$. Furthermore, $L'_{i-1}$ does not contain any parent of $v$ before the execution of the first such intermediate operation, hence at that time the parents of $v$ are in level $i-2$ or lower. Therefore,  any intermediate operation which is a $\procmove()$ or $\procdelete()$ keeps the invariant that the parents of $v$ are in level $i-2$ or lower as every $\procmove()$ operation can only decrease the level of a vertex. Any intermediate operation which is an $\procinsert()$ cannot add a parent of $v$ at all, since vertices are inserted in topological order. Hence $L'_{i-1}$ does not contain parents of $v$ at the time of the execution of $\procmove(v)$, a contradiction.

\smallskip\noindent (4) This case trivially follows by case (3) of the invariant.

\smallskip\noindent  (5) By induction, the invariant was satisfied right before the last execution of $\procinsert(v)$. We claim that after any operation the list $\MOVE$ contains all vertices for which this invariant is not satisfied. Right before the last $\procinsert(v)$ operation, $\MOVE = \emptyset$. The $\procinsert(v)$ operation simply adds $v$ either to layer $L'_{1}$ if there is no parent of $v$ in $G'$, or to layer $L_{h(v)}$ such that $L_{h(v)-1}$ contains a parent of $v$. Clearly, the  claim holds as $\MOVE = \emptyset$ and the invariant is satisfied for $v$. 
Consider now a $\procdelete(v)$ operation on a vertex $v\in L'_{i}$. The additional vertices that violate the invariant after this operation are the vertices $u \in L'_{i+1}$ whose only parent in $L'_{i}$ is $v$. Recall, we keep track of the parent of $u$ in $L'_{i}$ in the list $L_{\nextrm}(w)$. Since $delete(v)$ tests whether $L_{\nextrm}(u) = \emptyset$ after removing $v$ from $L'_i$ for all children $u$ of $v$ in $L'_{i+1}$, all these vertices are correctly inserted to $\MOVE$. Thus, the claim holds also after a $\procdelete(v)$ operation. Finally, we consider the case of a $\procmove(v)$ operation. The $\procmove(v)$ first removes a vertex $v$ from a list $L_i$. At the first stage of $\procmove(v)$, similarly to the $\procdelete(v)$ operation, the additional vertices that violate the invariant are the vertices $u \in L'_{i+1}$ whose only parent in $L'_{i}$ is $v$. Arguing in the exact same way as the $\procdelete(v)$ operation, all the additional vertices that violate the invariant (i.e., the ones that are not already in $\MOVE$) are correctly added to $\MOVE$. To complete our claim, notice that (similarly to $\procinsert()$) $\procmove(v)$ adds $v$ to either adds $v$ to layer $L'_{1}$ if there is no parent of $v$ in $G'$, or to $L'_{h(v)}$ such that $L'_{h(v)-1}$ contains a parent of $v$. Hence, the invariant is satisfied for $v$ after $\procmove(v)$, and therefore our claim holds. 
The proof of the invariant follows by the fact that $\MOVE=\emptyset$ right before an $\procinsert(v)$ operation.
\end{proof}

\begin{lemma}\label{lem:timeDecomposition}
The above algorithm (pseudo-code in Algorithm \ref{alg:decomposition}) takes $\bigo(n^2)$ time.
\end{lemma}
\begin{proof}
The running time is dominated by the execution of the operations $\procinsert()$, $\procdelete()$ and $\procmove()$. We execute $\procdelete(v)$ at most once on each $v\in V$. Assume $v\in L'_i$ at time of execution. This operation requires to remove $v$ from $|L_{\prevrm}(v)|$ lists $L_{\nextrm}(w)$ of vertices $w\in L'_{i+1}$: notice that we maintain pointers to the occurrence of $v$ in each of these lists $L_{\nextrm}(w)$, hence this operation can be performed in time $\bigo(\ell)$ since by Invariant \ref{inv:update} $|L_{\prevrm}(v)|\leq |L'_{i+1}|\leq \ell$. For $i\geq 1$, we also need to remove $v$ from the list $L_{\prevrm}(u)$ of some $u\in L_{i-1}$. The same invariant guarantees that $|L_{\prevrm}(u)|\leq |L'_{i}|\leq \ell$, and the fact that we store pointers of the occurrence of $v$ in each of these lists $L_{\prevrm}(u)$, and hence this step also takes $\bigo(\ell)$ time. Thus the total cost of $\procdelete()$ operations is $\bigo(n\ell)$. 

Similarly, $\procinsert(v)$ is executed at most once on each $v\in V$, and this operation can be easily performed in time $\bigo(n)$. Hence the total cost of  $\procinsert()$ operations is~$\bigo(n^2)$. 

It remains to consider the cost of $\procmove()$ operations. Let us focus on the operations of type $\procmove(v)$ for a specific vertex $v$ (notice that the same vertex $v$ can be moved multiple times). Assume $v\in L'_i$ at that time, and $v$ is moved to layer $L'_j$. Recall that by construction $j<i$. Similarly to the $\procdelete(v)$ case, we spend $\bigo(\ell)$ time to remove $v$ from affected lists $L_{\nextrm}(w)$, $w\in L'_{i+1}$, and $L_{\prevrm}(u)$, $u\in L'_{i-1}$. Analogously, we spend $\bigo(\ell)$ time to create the new list $L_{\nextrm}(v)$ and $\bigo(1)$ time to update $L_{\prevrm}(u)$ for some $u\in L'_{j-1}$. The rest of the operations can be easily performed in time $\bigo(\ell)$ for each level between $i+2$ and $j-1$.  Hence the cost of this operation $move(v)$ is $\bigo((j-i)\ell)$. Since the largest possible level of a vertex $v$ on which we execute $\procmove(v)$ is $h-1$, a simple  sum argument shows that the total cost of $\procmove(v)$ operations involving the same vertex $v$ is $\bigo(h \ell)=\bigo(n)$. Hence the total cost of  $\procmove()$ operations is $\bigo(n^2)$.  
\end{proof}

\section{Rectangular $\max,\min$-product}\label{sec:recangular-max-min-product}
Recall the definition of dominance product of two matrices $A,B$: $(A \ovee B)[i,j] = \max_k \min(A[i,k],B[k,j])$. In the following, we find useful the dominance product: $A \olessthan B$ as $(A \olessthan B)[i,j] = | \{k : A[i,k] < B[k,j] \} |$.

\begin{lemma}[c.f. Thm 3.1 in \cite{DBLP:conf/soda/DuanP09}]
\label{lem:sparsedominance}
If $A$ and $B$ are respectively $n \times p$ and $p \times n$ matrices with $m_1$ and $m_2$ non $(-\infty)$ elements, then $A \olessthan B$ can be computed in time $\widetilde\bigo(MM(n,p,n) + m_1m_2/p)$.
\end{lemma}
\begin{proof}
By reductions presented in \cite{DBLP:journals/siamcomp/WilliamsW13} (see \cite{GLU:2018} for alternative exposition), dominance product reduces to $\bigo(\log n)$ Hamming products with the reduction preserving dimensions and sparsity. 
We thus have to compute $\bigo(\log n)$ sparse Hamming products, with dimensions $n \times p$ and $p \times n$, and sparsity $m_1$ and $m_2$ respectively. 
By folklore reduction (see full version of \cite{GLU:2018} for exposition), each such product reduces to $n \times (np)$ vs $(np) \times n$ matrix product with sparsity $m_1$ and $m_2$ respectively. 
By techniques of \cite{DBLP:journals/talg/YusterZ05}, such product can be computed by decomposing into ``dense'' matrix product with cost $MM(n, p, n)$ (packing $p$ densest columns of first matrix, and $p$ densest rows of second matrix), and ``sparse'' matrix product with total cost $m_1 m_2 /p$.
\end{proof}
We provide following theorem for completeness (note, that we provide version with extra $\log$ factor, for the sake of shortening the proof).

\thmmaxminproduct*
\iffalse
\begin{theorem}[Corollary of \cite{DBLP:conf/soda/DuanP09}]
If $A$ and $B$ are respectively $n \times p$ and $p \times n$ matrices, then the $A \ovee B$ product can be computed in time $\widetilde\bigo(\sqrt{\textrm{MM}(n,p,n) \cdot n^2 p}\,)$.
\label{theorem:max-min-product-bound}
\end{theorem}
\fi
\begin{proof}
Let $L$ denote set of all the values in $A$ and $B$ of size $2np$ (w.l.o.g. all the values are distinct). We then partition $L$ into  $L_1, \ldots, L_t$, where each $L_r$ contains at most $\lceil 2np/t \rceil$ consecutive values from $L$.  We then construct sparse matrices $A_1, \ldots, A_t$ and $B_1, \ldots, B_t$ of dimensions $n \times p$ and $p \times n$, such that:
$$A_r[i,j] = \begin{cases}A[i,j] \text{ if } A[i,j] \in L_r\\ \infty \text{ otherwise} \end{cases}$$
$$B_r[i,j] = \begin{cases}B[i,j] \text{ if } B[i,j] \in L_r\\ -\infty \text{ otherwise} \end{cases}$$
For each $A_r$, a \emph{row-balancing} operation is applied (see Definition 2.1, \cite{DBLP:conf/soda/DuanP09}), producing $A'_r$ and $A''_r$, each of dimension $n \times p$ with $\bigo(p/t)$ elements in each row that are not $\infty$.

By construction from Theorem 3.3 in \cite{DBLP:conf/soda/DuanP09}, we need to compute, for each $r$: $A_r \olessthan B$, $A'_r \olessthan B$ and $A''_r \olessthan B$. Each such product reduces to: multiplication of Boolean matrices of dimension $n \times p$ with $p \times n$, and sparse dominance product of dimension $n \times p$ with $p \times n$ and density $m_1 = m_2 = \bigo(np/t)$. By Lemma~\ref{lem:sparsedominance}, this takes $\widetilde\bigo(MM(n,p,n) + n^2p/t^2)$ for each product. The postprocessing phase takes $\bigo(p/t)$ time for each $n^2$ elements of the output. The total time is then
$\widetilde\bigo(MM(n,p,n) t + n^2p/t)$, so setting $t = \sqrt{n^2 p / MM(n,p,n)}$ implies the runtime bound.
\end{proof}

\section{Faster decomposition in sparse graphs.}

We observe that our decomposition algorithm can be implemented more efficiently in sparse graphs (more precisely, whenever $m/\ell \ll n$). This is not critical in our application, since the number of edges will be $\Theta(n^2)$ in our case. However, since this might be helpful in other applications, we give the details in the following.
\begin{theorem}\label{thr:spaseDecomposition}
Let $G=(V,E)$ be a DAG with $n$ vertices and $m$ edges, represented via adjacency lists, and let $\ell\in [1,n]$ be an integer parameter. Then there exists an $\bigo(\frac{mn}{\ell})$ time deterministic algorithm to compute an $(\ell,\frac{2n}{\ell})$ decomposition of $G$.
\end{theorem}

\begin{proof}
We modify the above algorithm as follows. We do not compute the adjacency matrix of $G$, and we compute the topological order of $G$ in time $\bigo(m+n)$.  In each $\procinsert(v)$ operation, we simply scan the in-neighbors of $v$ and check in which layer they are to identify the layer where $v$ has to be inserted. We similarly modify the involved lists $L_{\nextrm}()$ and $L_{\prevrm}()$. Hence we can perform this operation in time $\bigo(\degree(v))$, where $\degree(v)$ is the degree of $v$ in $G$. Similarly, for each $\procdelete(v)$ operation, $v\in L'_i$, we consider the out-neighbors of $G$ and check which ones belong to $L'_{i+1}$. Hence also this operation can be performed in time $\bigo(\degree(v))$. Thus $\procinsert()$ and $\procdelete()$ operations cost $\bigo(m)$ in total. In each $\procmove(v)$ operation we consider all parents of $v$ and identify the lowest layer of any such parent. Hence this operation can be implemented in $\bigo(\degree(v))$ time. By construction each time we execute $\procmove(v)$, $v$ is moved to a strictly lower level. Since the largest level of a vertex $v$ on which we execute $\procmove(v)$ is $h-1$, we can perform this operation at most $h-2$ times. So the total cost of $\procmove()$ operations is $\bigo(\sum_{v\in V}\degree(v)\frac{n}{\ell})=\bigo(\frac{mn}{\ell})$. The claim follows. 
\end{proof}

\section{Faster LCA for smaller set of queries.}

We now show that if one is interested in computing the LCA for all pairs of vertices from a subset $S \subset V$ of the vertices, where $|S| = \bigo(n^\delta), \delta \le 1$, we can modify the algorithm from Section  \ref{sec:mainresult} to run faster. 
We refer to this problem as the \emph{$S$-pairs LCA} problem.

On a high-level, the algorithm remains the same. Let $G_{input}$ be the input DAG. We first compute in $\bigo(n^\omega)$ the transitive closure of $G$, and solve the $S$-pairs LCA problem on $G$.
Then, we compute in $\bigo(n^2)$ time an $(n^x,2n^{1-x})$-decomposition $(\cP,\cQ)$ of $G$ with the algorithm from Theorem \ref{thr:mainDecomposition}.
Again, the parameter $x$ will be fixed later on to optimize the running time of the algorithm.

\begin{algorithm}[t]
%\SetAlgoLined
\textbf{Input:} {Transitive closure graph $G=(V,E)$, a subset $S \subset V$ of vertices, 
and a family of chains $\mathcal{P}=\{P_1, \dots, P_p\}$ of $G$ where $p\leq n^{x}$.}\\
\textbf{Output:} {$\mathcal{P}$-restricted LCA $\lca_{\mathcal{P}}(u,v)$ for each pair of vertices $u,v\in S$.} 
    
    \begin{algorithmic}[1]
    \STATE Initialize $\lca_{\cP}(\cdot,\cdot)$ with $-\infty$\;
    
    \STATE Let $A$ be a $|S| \times p$ matrix\;
    
    \FOR{$P_i \in \mathcal{P}$}
        \FOR{$v\in S$}
            \STATE Let $w_i(v)$ be the parent of $v$ in $P_i$ with the largest index, otherwise $w_i(v)=-\infty$\;
            
            \STATE Set $A[v, i] \leftarrow w_i(v)$\;
        \ENDFOR
    \ENDFOR
    
    \STATE Compute the (max,min)-product $A \ovee A^T$\;

    \FOR{all $u,v \in S, u\not=v$}
            \STATE $\lca_{\mathcal{P}}(u,v) \leftarrow (A \ovee A^T)[u,v]$\;
    \ENDFOR

\end{algorithmic}
 \caption{Compute $\lca_{\mathcal{P}}(u,v)$ for all pairs of vertices $u,v\in S$.}
 \label{alg:chain-restricted-lca-small-set}
\end{algorithm}
Incorporating the $W$-restricted LCAs to the context of the $S$-pairs LCA problem, we give the following definition.

\begin{definition}
Given a DAG $G=(V,E)$, and two subsets of vertices $W,S\subseteq V$, the $W$-restricted $S$-Pairs LCA problem is to compute $\lca_W(u,v)\in LCA_W(u,v)$ for all pairs of vertices $u,v\in S$ ($\lca_W(u,v)=-\infty$ if $LCA_W(u,v)=\emptyset$).   
\end{definition}

Similarly to Section \ref{sec:mainresult}, we compute a solution to the $\cP$-restricted and $\cQ$-restricted $S$-Pairs LCA problems (that is, the values $\lca_{\cP}(u,v)$ and $\lca_{\cQ}(u,v)$ for all pairs of vertices $u,v\in S$), and later on combine these solutions to compute a solution to the $S$-pairs LCA problem. Again, we set $\lca(u,v)= \max \{\lca_{\cP}(u,v),\lca_{\cQ}(u,v)\}$, where the labels of the vertices respect a $\cQ$-compact topological order, where $\cQ$ is a path-respecting family of antichains of $G$.
Throughout the section, we assume the vertices are are labeled with integers $1, \dots, n$ according to a $\cQ$-compact topological order.

While the modifications to the algorithm from Section \ref{sec:mainresult} are straightforward, and the proof of correctness is essentially the same, we still spell-out the details for completeness.

The first modifications are in the algorithms that compute the solutions to the $\cP$-restricted and $\cQ$-restricted $S$-Pairs LCA problems. 
The modified version of Algorithm \ref{alg:chain-restricted-lca} is presented in Algorithm \ref{alg:chain-restricted-lca-small-set} and its proof in Lemma \ref{lemma:chain-restricted-lca-small-set}, while the modified version of Algorithm \ref{alg:antichain-restricted-lca} is presented in Algorithm \ref{alg:antichain-restricted-lca-small-set} and its proof in Lemma~\ref{lemma:antichain-restricted-lca-small-set}.

\begin{lemma}
Algorithm \ref{alg:chain-restricted-lca-small-set} computes for each pair of vertices $u,v \in S$, $|S|=n^{\delta}$, a $\mathcal{P}$-restricted LCA $\lca_{\mathcal{P}}(u,v)$. The algorithm runs in $\widetilde{\bigo}(n^{\frac{\omega(\delta,x,\delta) + 2\delta + x}{2}})$.
\label{lemma:chain-restricted-lca-small-set}
\end{lemma}
\begin{proof}
The proof of correctness follows from Lemma \ref{lemma:chain-restricted-lca}.
We now show the proof for the running time.
For each vertex $v$ this can be done in $\bigo(n)$ time for all $P_i$ by scanning all incoming edges from $v$, and in $\bigo(n^2)$ for all vertices in $S$ and all $P_i, 1 \leq i \leq p$.

The $(\max,\min)$ matrix multiplication $A \ovee A^T$, where $A$ is an $n^{\delta}\times n^{x}$, can be performed in time $\widetilde{\bigo}(n^{\frac{\omega(\delta,x,\delta) + 2\delta + x}{2}})$, by Theorem \ref{theorem:max-min-product-bound}. The time to compute the $(\max,\min)$ matrix product dominates the running time of Algorithm \ref{alg:chain-restricted-lca}.
\end{proof}

\begin{algorithm}[t]
%\SetAlgoLined
\textbf{Input:} {Transitive closure graph $G=(V,E)$, a subset $S \subset V$ of the vertices, and a family of antichains $\mathcal{Q}=\{Q_1, \dots, Q_q\}$ of $G$ that is path-respecting such that $q\leq 2n^{1-x}$.}\\
\textbf{Output:} {$\mathcal{Q}$-restricted LCA $\lca_{\mathcal{Q}}(u,v)$ for each pair of vertices $u,v\in V$.}

\begin{algorithmic}[1]
	
\STATE Initialize $\lca_{\cQ}(\cdot,\cdot)$ with $-\infty$\;

% Let $\phi_{V}:V\xrightarrow{1:1} \{1,\dots, |V|\}$ be an arbitrary bijection and $\phi^{-1}_{V}(\cdot)$ be its reverse function.
  \FOR{$i=q,\ldots,1$}
    \STATE Initialize a $|S| \times n^x$ matrix $A$ with zeros\;
    
 \STATE Let $\phi_{i}:Q_{i}\xrightarrow{1:1} \{1,\dots, |Q_{i}|\}$ be an arbitrary bijection and $\phi^{-1}_{i}(\cdot)$ be its inverse function\;
    
    \FOR{all $x\in |S|,y\in Q_{i}$ such that $(x,y)\in E$}
       \STATE  $A[x, \phi_{i}(y)] \leftarrow 1$\;
    \ENDFOR
    
    \STATE Compute $A\cdot A^T$, and its witness matrix $W$\;
    
    \FOR{all $u,v \in S, u\not=v$}
        \IF{$\lca_{\mathcal{Q}}(u,v) = -\infty$ and $A\cdot A^T[u,v] \not=0$}
            \STATE $\lca_{\mathcal{Q}}(u,v) \leftarrow \phi^{-1}_{i}(W[u,v])$\;
        \ENDIF
    \ENDFOR
    
 \ENDFOR
\end{algorithmic}

 \caption{Compute $\lca_{\mathcal{Q}}(u,v)$ for all pairs of vertices $u,v\in V$.}
 \label{alg:antichain-restricted-lca-small-set}
\end{algorithm}

\begin{lemma}
Algorithm \ref{alg:antichain-restricted-lca-small-set} computes the $\mathcal{Q}$-restricted LCA $lca_{\mathcal{Q}}(u,v)\in Q_i$ for each pair of vertices $u,v\in S \subseteq V, |S| = n^\delta$. The algorithm runs in time $\widetilde{\bigo}(n^{1-x+\omega(\delta,x,\delta)})$.
\label{lemma:antichain-restricted-lca-small-set}
\end{lemma}
\begin{proof}
The proof of correctness follows from Lemma \ref{lemma:antichain-restricted-lca-small-set}.
Initializing the $n^{\delta}\times n^{x}$ dimensional matrix $A$ for all $q=2n^{1-x}$ iterations take $\bigo(n^{\delta}n) \in \bigo(n^2)$ time.
We apply $n^{1-x}$ rectangular Boolean matrix multiplications $A\cdot A^T$ and witnesses.
The running time of the algorithm is upper bounded by $\widetilde\bigo(\sum_{i=1}^q \textrm{MM}(n^\delta,|Q_i|,n^\delta))$. 
Assume w.l.o.g. that $|Q_i|$ is non-increasing, then $|Q_i| \le n/i$, and by monotonicity of  $\textrm{MM}(n^\delta,\cdot,n^\delta)$
\begin{align*}
\sum_{i=1}^q \textrm{MM}(n^\delta,|Q_i|,n^\delta) &\le \sum_{i=1}^q \textrm{MM}(n^\delta,n/i,n^\delta)\\
&\le \sum_{j=0}^{\log q} 2^j \textrm{MM}(n^\delta,n/2^j,n^\delta)\\
&\le (1+\log q) \cdot q \cdot \textrm{MM}(n^\delta,n/q,n^\delta)\\
&\in \widetilde\bigo(n^{1-x+\omega(\delta,x,\delta)}).
\end{align*}
\end{proof}

\begin{algorithm}[t]
%\SetAlgoLined
\textbf{Input:} {DAG $G_{input}=(V,E_{input})$}\\
\textbf{Output:} {$\lca(u,v)$ for each pair of vertices $u,v\in S$}

\begin{algorithmic}[1]
    \STATE Compute the transitive closure graph $G=(V,E)$ of $G_{input}$\;

    \STATE Use Algorithm \ref{alg:decomposition}
    to compute a $(n^x, 2n^{1-x})$-decomposition into a family of chains $\mathcal{P}=\{P_1, \dots, P_p\}$ with $p \leq n^x$ and $|P_i| \leq n^{1-x}$ and a family of antichains $\mathcal{Q'} = \{Q'_1, \dots, Q'_{q'}\}$ with $q' \leq 2n^{1-x}$\;
    
    \STATE Use Lemma \ref{lemma:total-order-antichains} with input $\mathcal{Q'}$ to compute a path-respecting family of antichains $\mathcal{Q}=\{Q_1,\dots, Q_{q}\}$ of $G$ where $q \leq 2n^{1-x}$\;

    \STATE Compute a $\mathcal{Q}$-compact topological order of $G$ using Lemma \ref{lemma:q-compact-topological-order-computation} and rename vertices so that they are $1,\ldots,n$ according to this order\;

    \STATE Use Algorithm \ref{alg:chain-restricted-lca} to compute $\mathcal{P}$-restricted LCA $\lca_{\mathcal{P}}(u,v)$ for each pair of vertices $u,v\in S$\;

    \STATE Use Algorithm \ref{alg:antichain-restricted-lca} to compute $\mathcal{Q}$-restricted LCA $\lca_{\mathcal{Q}}(u,v)$ for each pair of vertices $u,v\in S$\;
    \FOR{\textbf{all} $u,v \in S, u\not=v$}
        \STATE $\lca(u,v) \leftarrow \max\{\lca_{\mathcal{Q}}(u,v), \lca_{\mathcal{P}}(u,v)\}$\;
    \ENDFOR   
    
\end{algorithmic}
 \caption{Compute $\lca(u,v)$ for all pairs of vertices $u,v\in S$}
 \label{alg:All-Pairs-global-lca-small-set}
\end{algorithm}

\begin{theorem}
Algorithm \ref{alg:All-Pairs-global-lca-small-set} computes for all pairs of vertices $u,v\in S$ a LCA $\lca(u,v)$. If $|S| = n^\delta$, then the algorithm runs in time $\bigo(n^\omega + n^{1-x+\omega(\delta,x,\delta)} + n^{\frac{\omega(\delta,x,\delta) + 2\delta + x}{2}})$.
\end{theorem}
\begin{proof}
Let $|S| = \bigo(n^\delta)$.
The proof for correctness follows from Theorem \ref{th:runtimebyx}. 
The running time of the algorithm is trivially $\widetilde\bigo(n^{\omega} + n^{\frac{\omega(\delta,x,\delta)+2\delta+x}{2}}+n^{1-x + \omega(\delta,x,\delta)})$ for a fixed $x\in [0,1]$. 
\end{proof}

What is now left is to find optimal value of $x$ as a function of $\delta$. Balancing the cost terms, we need to have $1-x+\omega(\delta,x,\delta) = \frac{\omega(\delta,x,\delta) + 2\delta + x}{2}$ which is equivalent to $\frac{2-2\delta}{\delta} +\omega(1,\frac{x}{\delta},1) = 3 \frac{x}{\delta}$. (Here we are using $\omega(at,bt,ct) = t \cdot \omega(a,b,c)$ property which holds for any $a,b,c,t \ge 0$.)
Using square matrix multiplication as a subroutine to implement rectangular matrix multiplication (i.e., the bound in \eqref{eq:boundsquare}), one obtains $x = \frac{2}{5-\omega}$ and $\gamma'=2\delta+\frac{4}{5-\omega}-1 \le 2 \delta + 0.522571$.
As usual, one can do better using more refined rectangular matrix multiplication algorithms. In particular, using the bound in \eqref{eq:boundrect}, one gets $x = \frac{2-\beta \alpha \delta}{3 - \beta}$, $\gamma' = \frac{2(2-\beta \alpha \delta)}{3 - \beta} + 2\delta-1 \le 0.6282973594 + 1.86112061279 \delta$.
If instead we apply the bound \eqref{eq:legallurrutia}, we get $\gamma' \le 0.711508 + 1.73504 \delta$ for $x = 0.855754 - 0.132478 \delta $. 

The running time becomes $\widetilde\bigo(n^\omega)$ for $|S| = n^\delta$ and: $\delta \le \frac{\omega+1}{2}-\frac{2}{5-\omega} < 0.925146$ if bound \eqref{eq:boundsquare} is used, $\delta \le  \frac{(\omega+1)(3-\beta) - 4}{2(3-\beta)- 2\beta \alpha} < 0.937374$ using the bound in \eqref{eq:boundrect} and $\delta \le 0.957531$ if we apply the bound \eqref{eq:legallurrutia} (given current bounds on $\omega(1,\cdot,1)$).

\section{Conclusions and Open Problems}

To the best of our knowledge, All-Pairs LCA is the first example of a natural graph problem with an algorithm based on fast matrix multiplication, which has a running time strictly between $\Omega(n^2)$ and $\bigo(n^{2.5})$, under the assumption $\omega=2$. 
This might suggest that a faster algorithm exists (e.g., with a running time of $\widetilde\bigo(n^{\omega})$). Alternatively, it would be interesting to derive fine-grained lower bounds based on All-Pairs LCAs in DAGs.

A simple greedy algorithm for decomposing a DAG into $\bigo(\sqrt{n}\,)$ chains and antichains runs in time $\bigo(n^{2.5})$ for dense graphs. Our algorithm improves this bound to $\bigo(n^2)$. In the similar problem of decomposing a sequence into  $\bigo(\sqrt{n}\,)$ monotonic subsequences, a naive greedy algorithm works in time $\bigo(n^{1.5} \log(n))$. Yehuda and Fogel improved this to $\bigo(n^{1.5})$ \cite{DBLP:journals/acta/Bar-YehudaF98} and there has been no further progress ever since.
It was also noted by J{\o}rgensen and Pettie that this is a natural example of a problem with a large ($\widetilde{\Omega} (\sqrt{n}\,)$) gap between the current algorithmic and decision-tree complexity \cite{DBLP:conf/focs/JorgensenP14}. 
Therefore, it would be interesting to see if the techniques  developed in this paper can be used to improve the time complexity of sequence decomposition. Alternatively, one could further investigate the relationship between the two problems in order to prove some % kind of 
lower bounds.

\subsection*{Acknowledgments}
The authors would like to thank Adam Polak for insightful discussions.

\bibliography{bib}
\end{document}